\documentclass[letterpaper, twocolumn, 10pt]{article}

\usepackage{usenix2019_v3}
\usepackage[]{hotpow}
\setrepo{https://github.com/pkel/hotpow/tree/arxiv_v3}

\begin{document}

\title{\theTitle}
\author{\theAuthorsLong}
\date{}

\maketitle

\begin{abstract}
  A fundamental conflict of many proof-of-work systems is that they want to
achieve inclusiveness and security at the same time. We analyze and resolve
this conflict with a theory of proof-of-work quorums, which enables a new
bridge between Byzantine and Nakamoto consensus.  The theory yields stochastic
uniqueness of quorums as a function of  a security parameter.  We employ the
theory in \mbox{\theProtocol{}}\ifanonprot{\footnote{Protocol name changed for
double-blind review.}}{}, a scalable permissionless distributed log protocol
that supports finality based on the pipelined three-phase commit previously
presented for HotStuff~\cite{yin2019HotStuffBFT}. 
We evaluate \theProtocol{} and variants with adversarial modifications by
simulation. Results show that the protocol can tolerate network latency, churn,
and targeted attacks on consistency and liveness with a small storage overhead
compared to plain Nakamoto consensus and less complexity than protocols that
rely on sidechains for finality.

\end{abstract}

\begin{figure*}
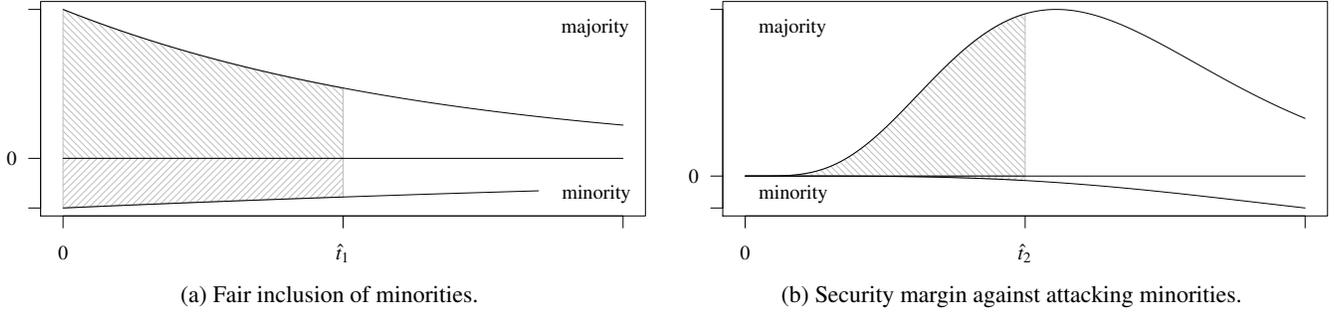

  \begin{subfigure}{0.490\linewidth}
    \resizebox{\linewidth}{!}{\input{figures/minority_vs_majority_exponential.tex}}
    \caption{Fair inclusion of minorities.}
    \label{fig:min_maj_exp}
  \end{subfigure}
  \hfill
  \begin{subfigure}{0.490\linewidth}
    \resizebox{\linewidth}{!}{\input{figures/minority_vs_majority_gamma.tex}}
    \caption{Security margin against attacking minorities.}
    \label{fig:min_maj_gamma}
  \end{subfigure}
  \caption{Probability densities of exponential (left) and gamma distributions
    (right) as functions over time for a $2/3$ majority and a $1/3$ minority
    (with flipped $y$-axis). The area under each curve represents the odds of
  winning a race.}
  \label{fig:min_maj}
\end{figure*}

\section{Introduction}

Bitcoin surprised scholars in distributed systems, as well as in security~\cite{bonneau2015SoKResearch}. Authors have
called the new composition of known concepts a ``sweet
spot''~\cite{tschorsch2016BitcoinTechnical} in the design space for protocols, and praised the
complex way the components are put together as a ``true leap of
insight''~\cite{narayanan2017BitcoinAcademic} of Nakamoto~\cite{nakamoto2008BitcoinPeertopeer}.
Likely the most intriguing part is the way Bitcoin uses proof-of-work puzzles to secure a distributed log.%

The role of proof-of-work in Nakamoto consensus can be contemplated in several ways. First and most intuitively, the computational puzzles can be interpreted as a rate limit on new identities, which discourage Sybil attacks~\cite{douceur2002SybilAttack} in a lottery for blocks and new coins. Second, proof-of-work can be conceived as a game-proof variant of a probabilistic back-off mechanism, as used in media access control in computer networks. It reduces the risk of collisions when many nodes concurrently seek write access to a shared medium, the ledger.
Proof-of-work has been formalized in cryptographic security models of Nakamoto consensus~\cite{garay2015BitcoinBackbone,pass2017AnalysisBlockchain}. However, we are not aware of work pointing out the fundamental conflict between inclusiveness and security inherent to the way proof-of-work is used in the known distributed log protocols. 

This conflict precludes reliable and fast commits. Arguably, it is the reason why practical protocols trade finality for eventual consistency. 
But the lack of finality limits the applicability for %
high-value transactions~\cite{bonneau2016WhyBuy,gervais2016SecurityPerformance}, a potential show-stopper discussed even beyond the technical community \cite{budish2018EconomicLimits,auer2019DoomsdayEconomics}.

We tackle this conflict directly, leading to a theory of proof-of-work quorums, which enables new ways of using proof-of-work in permissionless distributed log protocols. We propose one such protocol, \theProtocol{}, demonstrating that finality with reliable and short time to commit is possible. Specifically, we do not rely on sidechains, a tool used in the literature to stack Byzantine on top of Nakamoto consensus~\cite{kogias2016EnhancingBitcoin,pass2017HybridConsensus,pass2018ThunderellaBlockchains}. Sidechains can add finality and increase throughput at the price of increased complexity, overhead, and tricky issues in the synchronization between layers~\cite{kogias2016EnhancingBitcoin,eyal2016BitcoinNGScalable}.

The proposed protocol is inspired by two recent breakthroughs:
Bobtail~\cite{bissias2020BobtailImproved} and HotStuff~\cite{yin2019HotStuffBFT}. The former
optimizes stochastic properties of the block delay in Nakamoto consensus. 
The latter adapts principles of Byzantine
fault tolerance to blockchains in a clever way. It has received attention after Facebook's announcement to use it in LibraBFT~\cite{calibra2019librabft}.

We make the following contributions:
\begin{enumerate}
  \item We draw attention to a fundamental conflict between inclusiveness and
    security in \nc{} and propose a principled resolution
    (Section~\ref{sec:intuition}).
  \item We develop a theory of proof-of-work quorums where quorums
    are formed over votes generated by stochastic processes.
    We show that sufficiently large quorums are practically unique (Section~\ref{sec:pow_quorum}).
  \item We propose \theProtocol{}, a protocol that finds
    consensus over a distributed log without requiring pre-defined identities.
    \theProtocol{} scales at least as well as practical blockchain
    protocols and much better than Byzantine fault tolerance protocols.
    It relies on proof-of-work, but, unlike deployed systems using the
    longest chain rule, our construction supports a three-phase commit logic.
    State updates (transactions) are final after a predictable amount of time,
    and the probability of inconsistency is bounded according to our theory
    (Section~\ref{sec:protocol}).
  \item We simulate executions of \theProtocol{} as well as of variants with adversarial
    modifications. The results show that the protocol can tolerate
    network latency, churn, and targeted attacks on consistency and liveness at
    small %
    overhead compared to the best deployed systems (Section~\ref{sec:evaluation}).
\end{enumerate}
Section~\ref{sec:discussion} compares \theProtocol{} to related works and
discusses its limitations. Section~\ref{sec:conclusion} concludes.
For replicability and future research, we make the protocol implementation and the simulation code available online.\repofootnote{}

\section{Intuition} \label{sec:intuition}

The key conflict between inclusiveness and security faced by cryptocurrencies is as follows:
\emph{minorities should be encouraged to participate (inclusiveness),
but they should not be able to make decisions alone (security).} \nc{}
achieves inclusiveness by  sacrificing security for an uncertain period of
time (eventual consistency).
This becomes problematic when irreversible real-world actions are taken based
on unsettled transactions in the distributed log (double spending). A short and reliable time to commit would mitigate this risk.

Recall that \nc{} prioritizes inclusiveness by using  a puzzle as gatekeeper to
participation.  The protocol specifies a repeated race for the first puzzle solution.
Each winner proposes a state update and receives some reward.  Most
cryptocurrencies use puzzles---moderately hard functions---for which iterative trial and error is the best
known solving algorithm. Such puzzles imply exponentially distributed solving
time.  Figure~\ref{fig:min_maj_exp} shows the probability distributions for the
solving times of a $2/3$ majority of solving power compared to a $1/3$
minority. The expected time of the end of the race is marked with $\hat{t}_1$
(in Bitcoin $\hat{t}_1 \approx 10$ minutes). Consequently, the area under each
curve represents the odds of winning the race. Observe that the minority has a
fair chance. This makes the protocol inclusive, but also implies that
minorities have a significant chance of directly writing state updates. For
improved security, we would prefer a distribution such that the minority's area
under the curve is small (ideally negligible), as displayed in
Figure~\ref{fig:min_maj_gamma}.

Since the puzzle of \nc{} behaves like in Figure~\ref{fig:min_maj_exp}, a
single state update is not reliable. As a result, users are recommended to wait
for multiple consecutive blocks before acting upon a payment. The time needed for
sequentially solving $k$ exponential puzzles is gamma distributed with shape
parameter~$k$.  In fact, Figure~\ref{fig:min_maj_gamma} shows the %
gamma distribution for $k = 6$. Note the significant gap between minority and majority: it is unlikely that a minority
can generate a sequence of $6$~state updates before the majority does so. In
this sense, multiple puzzle solutions qualify a majority, while a single one does not.

In \nc{}, security comes at the price of waiting for multiple
solutions. Bitcoin's convention of $k=6$ implies an expected waiting time of
$\hat{t}_2 \approx 60$ minutes, which is arguably too slow for many
applications. Besides, \nc{} does not give a rationale on how to choose $k$.

A key idea for resolving this conflict is to break the one-to-one relationship
between puzzle solutions and blocks.  Instead of requiring a single $10$ minute
puzzle per block, \theProtocol{} asks for $k$ easier puzzles each expected to
take $10/k$ minutes.
In other words, \theProtocol{} achieves security by appending puzzle solutions
\emph{in parallel} rather
than sequentially, as illustrated in Figure~\ref{fig:seq_vs_par}. Since the
puzzles are independent, we end up with the same block rate but $k$ times
the number of solutions. The expected computational effort stays the same, but
we accumulate a qualifying number of solutions for \emph{every} block. This means we
get the shape of Figure~\ref{fig:min_maj_gamma} much faster: $\hat{t}_2 \approx
\hat{t}_1$.

\begin{figure}
  \centering
  \resizebox{0.9\linewidth}{!}{\includegraphics{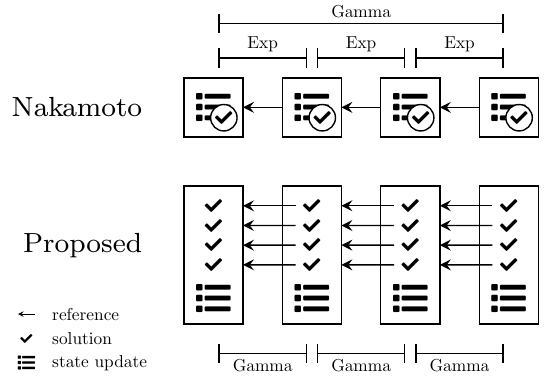}}
  \caption{Sequential puzzles in \nc{} imply exponentially distributed
    block delays. Multiple smaller puzzles in parallel lead to a
  gamma distribution for each block.}
  \label{fig:seq_vs_par}
\end{figure}

For a principled construction of \theProtocol{}, we reduce the payload ``authenticated''~\cite{back2014EnablingBlockchain} by proof-of-work to a minimum: %
\begin{enumerate}
  \item a reference to a recent point in time (\eg{} a hash link to the last seen block)
  \item a reference to an identity (public key or commitment)
\end{enumerate}
A triple of a puzzle solution and these two references forms a verifiable
ephemeral identity.
The puzzle solution binds resources in order to prevent Sybil attacks,
the reference in time ensures freshness, and
the identifier enables authorized actions, such as claiming a reward.

The main difference between proof-of-work systems and the well-studied class of
Byzantine fault tolerant (BFT) systems~%
\cite{lamport1982ByzantineGenerals, dwork1988ConsensusPresence, castro2002PracticalByzantine}
is that the former do not rely on
external identification of the participating nodes.
Inspired by the early work of \citet{aspnes2005Exposingcomputationallychallenged}, \theProtocol{} uses proof-of-work to bootstrap ephemeral identities and plugs them
into HotStuff~\cite{yin2019HotStuffBFT}, a state of the art blockchain-based BFT system.
In HotStuff, each block carries a certificate about a
qualified majority of nodes (quorum) confirming the last seen block.
HotStuff's proof of finality is based on the qualifying properties of each
quorum. This motivates us to explore whether and to what extent a set of
proof-of-work solutions can qualify a majority. In Section~%
\ref{sec:pow_quorum}, we will show that qualifying majorities are possible
within a single block. This allows us to transfer HotStuff's finality to the
permissionless setting.

The recurse to HotStuff enables us to fix the number of blocks to wait before
accepting a state update as final at the necessary number of phases to commit,
thereby resolving a drawback of \nc{}. As illustrated in
Figure~\ref{fig:pipeline}, HotStuff uses a three-phase commit, which can be
pipelined for subsequent state updates on a blockchain. In a nutshell, the
first phase locks a single proposal, the second phase confirms majority uptake
of this lock, and the third phase ensures that the knowledge of this knowledge
is propagated. We refer to~\cite{yin2019HotStuffBFT} for the rationales and
failure modes. In this sense, \theProtocol{} parallelizes not only puzzle solutions but
also the phases of the commit logic.

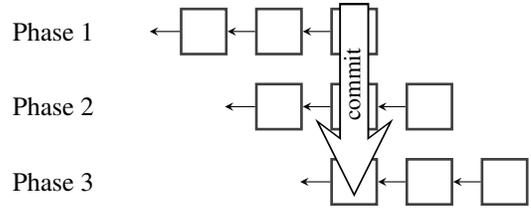
\begin{figure}
  \begin{center}
    \begin{tikzpicture}[>=stealth]
      \begin{scope}[color=darkgray, line width=1pt, minimum width=6mm,minimum height=6mm]

        \draw (0,0) node [draw] (a1) {};
        \draw (1,0) node [draw] (a2) {};
        \draw (2,0) node [draw] (a3) {};

        \draw (1,-1) node [draw] (b1) {};
        \draw (2,-1) node [draw] (b2) {};
        \draw (3,-1) node [draw] (b3) {};

        \draw (2,-2) node [draw] (c1) {};
        \draw (3,-2) node [draw] (c2) {};
        \draw (4,-2) node [draw] (c3) {};

      \end{scope}

      \foreach \c in {a,b,c}
      { \draw [->] (\c2.west)--(\c1.east);
        \draw [->] (\c3.west)--(\c2.east);
        \draw [->] (\c1.west)--++(-4mm,0);
      }

      \draw (-2, 0) node {Phase 1};
      \draw (-2,-1) node {Phase 2};
      \draw (-2,-2) node {Phase 3};

      \draw [line width=4mm,->] (2,.38) -- (2,-2.2);
      \draw [color=white,line width=3.5mm,->] (2,.35) -- (2,-2.15);
      \node [rotate=90] at (1.975,-0.7) {\small commit};

    \end{tikzpicture}
  \end{center}
  \vspace{-2ex}
  \caption{Pipelined three-phase commit on a blockchain in HotStuff and
  \theProtocol{}.}
  \label{fig:pipeline}
\end{figure}

Another advantage of the gamma distribution per block is a reduction
in the variance of block delays compared to the exponential distribution
implied by the puzzle.  While the commit pipeline gives us finality after three blocks, the reduced variance translates this into a reliable time to commit.
The theory in the following section shows formally how all this is related to the
quorum size, \theProtocol{}'s new security parameter.

\section{Proof-of-Work Quorums} \label{sec:pow_quorum}

Quorums are central to the design and analysis of BFT protocols. 
The typical Byzantine setting assumes a set of $n =
3f + 1$ identified nodes, of which at most $f$ deviate from the protocol. A set of $2f + 1$ votes for the same value is
called a quorum. If correct nodes vote at most once, quorums imply a majority decision and thus are unique.
The uniqueness may be violated in two situations.

\begin{enumerate}[label=BFT-\theenumi, wide=0pt, leftmargin=*]
  \item More than $n$ nodes vote. \label{bft-network}
  \item More than $f$ nodes vote more than once. \label{bft-adversary}
\end{enumerate}
Practical systems avoid~\ref{bft-network} using preset identities
for all nodes and rule out~\ref{bft-adversary} by assumption.

Proof-of-work enables systems where
agents can join and leave at any time without obtaining permission from an
identity provider or gatekeeper~\cite{nakamoto2008BitcoinPeertopeer}.  This difference is
often implied in the terms ``permissioned'' and ``permissionless''.
In the permissionless case one must distinguish between \emph{agents} and \emph{nodes}. Agents are entities participating in a distributed system. An agent can
operate any number of nodes.  Colluding parties are interpreted as a single agent. %

We introduce the notion \emph{proof-of-work quorum} for a set of votes where each vote
requires a solution to a proof-of-work puzzle. Since the puzzle solving time is probabilistic, the uniqueness of quorums cannot be absolute. In
contrast to the Byzantine setting, we have to consider three failure modes:

\begin{enumerate}[label=PoW-\theenumi, wide=0pt, leftmargin=*]
  \item The total compute power of the network is higher than assumed.
    \label{pow-network}
  \item The adversary controls more than the assumed fraction of compute power.
    \label{pow-adversary}
  \item  A random bad realization happens.
    \label{pow-probability}
\end{enumerate}

The failure modes \ref{pow-network} and \ref{pow-adversary} correspond to the
Byzantine failure modes \ref{bft-network} and \ref{bft-adversary}. Our goal is
to understand the new failure mode \ref{pow-probability} and how it affects the
potential ambiguity (violation of uniqueness) of quorums.

\begin{definition}[Proof-of-work process] \label{def:process}
  A proof-of-work process is a stochastic count process where each event
  assigns one \emph{ability to vote} (ATV) to one agent. Each ATV can be used
  by the agent it is assigned to, to vote once for one value.
\end{definition}

We adopt the notion of a quorum from the BFT literature~\cite{malkhi1998ByzantineQuorum, yin2019HotStuffBFT} except that we will apply it to votes from ATVs rather than identified nodes.

\begin{definition}[$\qsize$-quorum] \label{def:quorum}
  A set of $\qsize$ votes for the same value~$x$ is called a $\qsize$-quorum for~$x$.
\end{definition}

Observing a $\qsize$-quorum implies that at least $\qsize$~ATVs have been used, hence the
proof-of-work process must have assigned at least~$\qsize$ ATVs. This connects to time.

\begin{definition}[Optimistic quorum time] \label{def:oqt}
  The time at which the proof-of-work process assigns the $\qsize$-th ATV is
  called optimistic $\qsize$-quorum time.  For a proof-of-work process $P$
  and quorum size $\qsize$ it is formally defined by the random variable
  \[
    T_{P,\qsize} := \inf\{t \in \realsgez \mid P(t) \geq \qsize\} \,.
  \]
\end{definition}

$T_{P,\qsize}$ is the earliest point in time at which a $\qsize$-quorum
is feasible.  A $\qsize$-quorum is only possible at exactly $T_{P,\qsize}$, if all
assigned ATVs are used to vote for the same value. 

\newcommand{\poa}{\ensuremath{\operatorname{poa}}}

A quorum for $x$ is ambiguous if there is another quorum for $y \neq x$.
Since each ATV can be used for at most one value, ambiguous $\qsize$-quorums
are only possible when the proof-of-work process has assigned at least $2\qsize$ ATVs.

\begin{definition}[Probability of ambiguity] \label{def:poa}
  For a proof-of-work process $P$ and quorum size~\qsize{} we define the
  \emph{probability of ambiguity} (POA) as
  \[
    \poa_{P, \qsize}(t) := \prob{P(t) \geq 2 \qsize} \,.
  \]
\end{definition}

For puzzles where the best known solving algorithm is independent trial and error, the stochastic process is instantiated by the Poisson process $P_\lambda$.
This is because if each puzzle solution generates one ATV,
the time between consecutive ATVs is exponentially distributed with rate
$\lambda$.

\begin{lemma} \label{lem:poa_poisson}
  The POA for the Poisson process $P_\lambda$  is given by
  \[
    \poa_{P_\lambda, \qsize}(t) =
    1 - e^{-\lambda t} \sum_{i=0}^{2\qsize -1}{\frac{(\lambda t)^i}{i!}} \,.
  \]
\end{lemma}

\begin{proof} See Appendix~\ref{apx:proofs}.\end{proof}

\begin{lemma} \label{lem:qt_possion}
  The optimistic \qsize{}-quorum time for the Poisson process is Erlang
  distributed with shape parameter~\qsize{} and rate parameter~$\lambda$, in
  short
  \[
    T_{P_\lambda,\qsize} \drawn \distErlang(\qsize, \lambda) \,.
  \]
\end{lemma}

\begin{proof} See Appendix~\ref{apx:proofs}.\end{proof}

\newcommand{\tEv}{\ensuremath{\bar{t}_{\lambda, \qsize}}}

\begin{corollary}\label{cor:tEv}
  The expected optimistic $\qsize$-quorum time for the Poisson process is
  \[
    \tEv := \ev{T_{P_\lambda,\qsize}} = \qsize / \lambda\,.
  \]
\end{corollary}

\begin{proof}
  The statement follows from Lemma~\ref{lem:qt_possion} and the definition of
  the Erlang distribution~\cite[p.~146]{stewart2009ProbabilityMarkov}.
\end{proof}

Figure~\ref{fig:opt_qtime} illustrates the distribution of the optimistic
$\qsize$-quorum time for $\qsize \in \{1,2,16\}$ based on the Poisson process.
In order to compare quorum sizes greater than one to an ideal Bitcoin ($\qsize=1$, $\tEv{}=10$ minutes), we
choose $\lambda = \qsize / 10$.

\begin{figure}
  \centering
  \resizebox{\columnwidth}{!}{\input{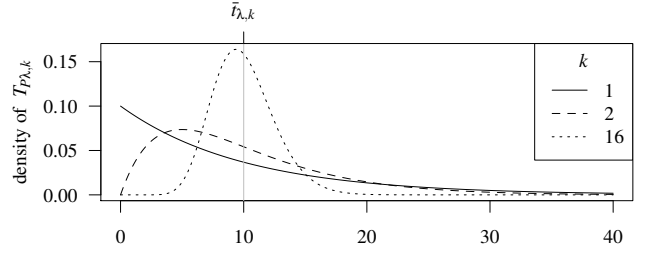}}
  \caption{The density of the distribution of the optimistic $\qsize$-quorum time
    based on $P_\lambda$ with rate $\lambda= \qsize/10$
  (minutes).}
  \label{fig:opt_qtime}
\end{figure}

Figure~\ref{fig:poa_over_time} shows the POA for different quorum sizes as a
function of time. Again, we adjust the rate such that the expected optimistic \qsize{}-quorum time is 10 minutes. Observe that
the POA increases over time as the number of ATVs grows. More
importantly, the POA at the expected optimistic quorum time decreases in the
quorum size~\qsize{}.

\begin{figure}
  \centering
  \resizebox{\columnwidth}{!}{\input{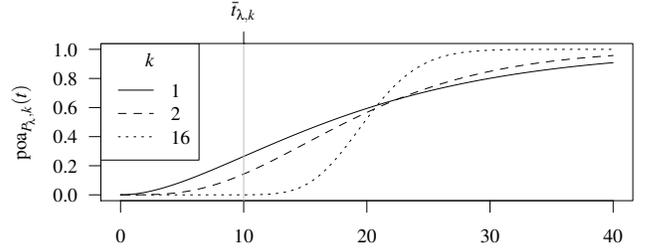}}
  \caption{The probability of ambiguity as a function of time for quorum
    sizes $\qsize = 1, 2, $ and $16$ and $\lambda = \qsize/10$ (minutes).}
  \label{fig:poa_over_time}
\end{figure}

In order to isolate the effect of \qsize{}, we evaluate the POA at fixed time
\tEv{}, which lends itself to a closed form.

\begin{corollary}\label{cor:poa_at_ev}
  For the Poisson process, the POA at expected optimistic \qsize{}-quorum time
  is given by
  \[
    \poa_{P_\lambda, \qsize}(\tEv) =
    1 - e^{-\qsize} \sum_{i=0}^{2\qsize -1}{\frac{\qsize^i}{i!}} \,.
  \]
\end{corollary}

\begin{proof}
  By inserting Corollary~\ref{cor:tEv} into Lemma~\ref{lem:poa_poisson}.
\end{proof}

Observe that the POA at expected optimistic quorum time is independent of
$\lambda$. This is useful as $\lambda$ may measure the total compute capacity
in proof-of-work networks, which is not necessarily known to each agent.

Since ambiguity causes failure, and the probability of ambiguity vanishes as
$\qsize$ grows, $\qsize$ becomes a security parameter.
In order to relate it to other security parameters, such as the key size, we
adopt the common definition of negligibility from cryptography (\ie{}
asymptotic decline faster than any polynomial) and state the following
theorem.

\begin{theorem} \label{thm:negligible}
  For the Poisson process, the probability of ambiguity at the expected quorum
  time is negligible in the quorum size~$\qsize$.
\end{theorem}

\begin{proof} See Appendix~\ref{apx:proofs}.\end{proof}

\begin{remark}[Validation on Bitcoin]
For Bitcoin parameters ($\qsize=1, \lambda=0.1$), the POA at \tEv{} is $p
= 0.2642$. This part of the theory can be validated on historical data.  We
estimate the expected block delay by averaging the differences between
consecutive block time stamps over 2017--2018.\footnote{We choose this time range because the block
time stamps were less accurate in the more distant past as the data field
was used for other purposes.} The estimated average block delay is $\hat{t}=9.52$ minutes.  The
ratio of cases with more than two blocks arriving within $\hat{t}$ is $\hat{p} = 0.2606$. 
This estimate should be slightly below $p$ because our historic data does not contain
orphaned blocks. Since $p \approx \hat{p}$, we conclude that the theory applies to Bitcoin.
\end{remark}

The implication of this theory for protocol design is that larger quorums
reduce the probability of ambiguity. 
The (close to) exponential decay makes it conceivable to choose parameters such that quorums are practically unique. This allows us to use a notion of quorum uniqueness with ephemeral identities generated by proof-of-work.

\section{\theProtocol{}} \label{sec:protocol}

\newcommand{\loghash}{\ensuremath{\mathcal{H}_{\text{list}}}}
\newcommand{\powhash}{\ensuremath{\mathcal{H}_{\text{pow}}}}
\newcommand{\initial}{\ensuremath{S_0}}
\newcommand{\opRef}{\ensuremath{r}}
\newcommand{\opId}{\ensuremath{p}}
\newcommand{\opSol}{\ensuremath{s}}
\newcommand{\votenp}{\ensuremath{\opRef,\opId,\opSol}}
\newcommand{\vote}{\ensuremath{(\votenp)}}

Now we specify \theProtocol{}, a distributed log protocol secured by a
proof-of-work process (Def.~\ref{def:process}) and $\qsize$-quorums
(Def.~\ref{def:quorum}).

We present \theProtocol{} using pseudo\-code and a mixture of event-driven and
imperative programming. %
A less ambiguous implementation in OCaml is provided online.\repofootnote

\subsection{Prerequisites} \label{ssec:proto_prerequisites}

We assume interfaces to the network and application layers (Fig.~\ref{fig:appstack}), and the availability of cryptographic primitives.

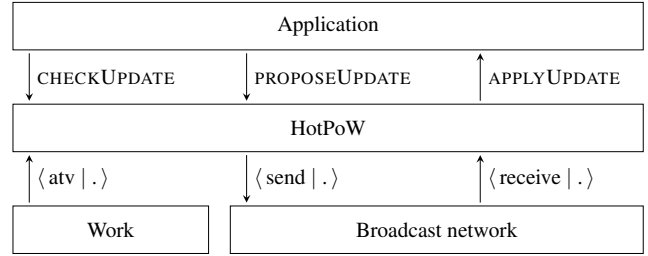
\begin{figure}
  \centering
  \resizebox{\linewidth}{!}{
    \small
    \begin{tikzpicture}[>=stealth, x=1.6cm, y=1cm]
      \draw (0.1,0.15) rectangle (1.9,0.85);
      \node at (1.00, 0.5) {Work};

      \draw (2.1,0.15) rectangle (5.9,0.85);
      \node at (4.00, 0.5) {Broadcast network};

      \draw (0.1,1.65) rectangle (5.9,2.35);
      \node at (3, 2) {\theProtocol{}};

      \draw (0.1,3.15) rectangle (5.9,3.85);
      \node at (3, 3.5) {Application};

      \path [<-] (0.25,2.4) edge node[anchor=west] {\Call{checkUpdate}{}}   +(0, 0.7);
      \path [<-] (2.25,2.4) edge node[anchor=west] {\Call{proposeUpdate}{}} +(0, 0.7);
      \path [->] (4.40,2.4) edge node[anchor=west] {\Call{applyUpdate}{}}   +(0, 0.7);
      \path [->] (0.25,0.9) edge node[anchor=west] {\event{atv $\mid$ .}}     +(0, 0.7);
      \path [<-] (2.25,0.9) edge node[anchor=west] {\event{send $\mid$ .}}    +(0, 0.7);
      \path [->] (4.40,0.9) edge node[anchor=west] {\event{receive $\mid$ .}} +(0, 0.7);

  \end{tikzpicture}}
  \caption{Interaction between the protocol (\theProtocol{}), the application, the
    proof-of-work process, and the network. Arrows denote information flows and
  not necessarily call directions.}
  \label{fig:appstack}
\end{figure}

\subsubsection{Broadcast Network} \label{sssec:method_net}
The proposed protocol requires a (potentially unreliable) network broadcast.
We abstract from the exact implementation and assume that scheduling an event 
\event{send $\mid$ $m$}  results in the message $m$ being sent to (most
of) the other \theProtocol{} nodes. On the receiving side, the implementation delivers
message $m'$ by scheduling \event{receive $\mid$ $m'$}.

\subsubsection{Application} \label{sssec:method_app}
\theProtocol{} implements a distributed log which may serve as a base for different
applications~\cite{lamport1978TimeClocks, schneider1990ImplementingFaulttolerant,
abraham2017BlockchainConsensus}.
For example, a simple cryptocurrency could append lists of transactions which jointly form a ledger.
More advanced applications could add scalability layers that only
record key decisions in the distributed log while handling other state updates
separately \cite{eyal2016BitcoinNGScalable,kogias2016EnhancingBitcoin,pass2018ThunderellaBlockchains}.

We abstract from the application logic using three procedures \theProtocol{} can call.
\Call{checkUpdate}{} takes an application state and a state update
    as arguments and returns true if the state update is valid.
\Call{applyUpdate}{} takes an application state and a state
    update and returns an updated state.
\Call{proposeUpdate}{} takes an application state  and
    returns a valid state update.
We are agnostic about direct access of the application to the broadcast network. 
For example, cryptocurrencies share transactions provisionally before they are logged in blocks.

\subsubsection{Cryptography} \label{ssec:method_dsa} \label{ssec:method_hash}
\theProtocol{} uses cryptographic hash functions for the hash-linked list and the
proof-of-work process. We separate these two concerns and use two different
hash functions, \loghash{} and \powhash{}.  While it is sufficient that
\loghash{} is cryptographically secure, \theProtocol{} requires the same stronger assumptions for \powhash{} as Bitcoin~\cite{abraham2017BlockchainConsensus}. 
Since this difference is not central, the reader can safely assume $\loghash{} = \powhash{} = \operatorname{SHA3}$.

\theProtocol{} also requires a digital signature scheme~\cite[Def.~12.1,
p.~442]{katz2014IntroductionModern}. We assume a secure implementation is given
by the three procedures \Call{generateKeyPair}{}, \Call{checkSignature}{}, and
\Call{sign}{}.  Every node holds an asymmetric key pair (me, secret).

\subsection{Protocol}

\subsubsection{Local Block Store} \label{sssec:proto_global}

\theProtocol{} nodes maintain a local tree of hash-linked blocks and a reference to the preferred chain
(head). They store blocks together with the associated application state, the
block height, and a set of corresponding votes (see Listing~\ref{lst:store}). The
block storage is indexed by \loghash{}.

\begin{listing}[H]
  \caption{Local Block Store}
  \label{lst:store}
  \begin{algorithmic}[1]
    \Procedure{store}{block B}
    \State h $\gets \loghash(\text{B})$
    \If {h $\not\in$ blocks}
    \State parent $\gets$ blocks[B.parent]
    \State blocks[h].parent $\gets$ parent
    \State blocks[h].state $\gets$ \Call{applyUpdate}{parent.state, B.payload}
    \State blocks[h].height $\gets$ parent.height + 1
    \State blocks[h].votes $\gets \emptyset$
    \State blocks[h].block $\gets$ B
    \State \Call{updateHead}{h}
    \EndIf
    \EndProcedure
    \algstore{hotpow}
  \end{algorithmic}
\end{listing}

\subsubsection{Votes} \label{sssec:proto_vote}
As mentioned in Section~\ref{sec:intuition}, a vote in \theProtocol{} is a
triple \vote{}, where \opRef{} is a reference to a previous block, \opId{} is
the public key of the voter, and \opSol{} is a puzzle solution.
A vote \vote{} is valid if $\powhash\vote\leq \vthres$, where \vthres{}
denotes the proof-of-work threshold and represents \theProtocol{}'s difficulty
parameter.
\theProtocol{} nodes maintain a set of valid votes for each block.
The procedure \Call{collect}{} (Listing~\ref{lst:collect}) adds a valid
vote \vote{} to the block referenced by \opRef{} and, if necessary, updates the
preferred chain (see Sect.~\ref{sssec:proto_preference} below).

\begin{listing}[H]
  \caption{Collection of Votes}
  \label{lst:collect}
  \begin{algorithmic}[1]
    \algrestore{hotpow}
    \Procedure{collect}{\votenp{}}
      \If {$\powhash\vote \leq \vthres$}
        \State blocks[\opRef{}].votes $\gets$ blocks[\opRef{}].votes $\cup \{\text{(\opId{}, \opSol{})}\}$
        \State \Call{updateHead}{r}
      \EndIf
    \EndProcedure
    \algstore{hotpow}
  \end{algorithmic}
\end{listing}

\subsubsection{Quorums} \label{sssec:proto_quorum}
As defined in Section~\ref{sec:pow_quorum}, a \qsize{}-quorum is a set
of~\qsize{} votes for the same reference. We represent such quorums as lists.
Since the reference is the same for all votes, we omit it from the list. %
A list $L = \{(\opId_i, \opSol_i)\}$ represents a valid \qsize{}-quorum
for $\opRef$, if the following conditions hold:
\begin{enumerate}
  \item \label{qcond_size}$|L| = \qsize$
  \item \label{qcond_threshold} $\forall\, 1 \leq i \leq \qsize \colon {\powhash(\opRef, \opId_i, \opSol_i)} \leq \vthres$
  \item \label{qcond_order} $\forall\, 1 \leq i < \qsize \colon \powhash(\opRef, \opId_i, \opSol_i) \leq \powhash(\opRef, \opId_{i+1}, \opSol_{i+1})$
\end{enumerate}
The first condition enforces the quorum size.
The second condition ensures that all votes are valid.
The third condition imposes a canonical order which we use for leader election.
We intentionally allow single nodes providing multiple votes. Sibyl attacks are
mitigated by the scarcity of votes.

\subsubsection{Leader Election} \label{sssec:proto_leader}

A quorum can only be formed at optimistic quorum time (Def.~\ref{def:oqt})
if all nodes vote for the same block. We facilitate coordination by electing
a leader who is responsible for proposing a new block. This election is based
on the proof-of-work quorum: the leader is identified by the smallest vote.
According to Section~\ref{sssec:proto_quorum} Condition~\ref{qcond_order}, this
vote is also the first element of the quorum.
Leaders  authenticate their proposals for the next block using \Call{sign}{}
and their private key.
Everyone verifies proposals with the first public key in the quorum.

\subsubsection{Blockchain} \label{sssec:proto_block}

The global data structure of the protocol is a hash-linked list of blocks.
Each block consists of a hash reference to its predecessor (parent), a
proof-of-work quorum for this predecessor, a payload, and a proof of leadership
(signature).
The references to parent blocks are established by the collision-resistant hash
function \loghash{}.
The payload is a state update to the application
implemented on top of the distributed log (see Sect.~\ref{sssec:method_app}).

With quorums, leader election, and state updates defined, we are in the
position to present \theProtocol's block validity rule in
Listing~\ref{lst:valid_block}.
The loop iterates over the quorum, counts the votes, verifies them, and checks
their canonical order.
The boolean conjunction in line~\ref{l:threecond} verifies the remaining
condition of the quorum, leadership, and the validity of the proposed state
update.

\begin{listing}[H]
  \caption{Block Validity}
  \label{lst:valid_block}
  \begin{algorithmic}
    \algrestore{hotpow}
    \Procedure{validBlock}{block B}
      \State $(c,h)  \gets (0,0)$
      \ForAll{$(\opId, \opSol)$ \textbf{in} B.quorum}
        \State $h' \gets \powhash(\text{B.parent}, \opId, \opSol)$
        \label{l:predecessor}
        \Comment{predecessor!}
        \If{$h' > \vthres$} \Return false
        \Comment{quorum condition~\ref{qcond_threshold}, Sect.~\ref{sssec:proto_quorum}}
        \EndIf
        \If{$h' < h$} \Return false
        \Comment{quorum condition~\ref{qcond_order}, Sect.~\ref{sssec:proto_quorum}}
        \EndIf
        \State $(c,h)  \gets (c + 1, h')$
      \EndFor
      \State \Return \label{l:threecond}
      \Statex \hskip\algorithmicindent \hskip\algorithmicindent
        $c = \qsize$ $\wedge$
        \Comment{quorum condition~\ref{qcond_size}, Sect.~\ref{sssec:proto_quorum}}
      \Statex \hskip\algorithmicindent \hskip\algorithmicindent
        \Call{checkSignature}{B.quorum.[0].\opId{}, B} $\wedge$
      \Statex \hskip\algorithmicindent \hskip\algorithmicindent
        \Call{checkUpdate}{blocks[B.parent].state, B.payload}
    \EndProcedure
    \algstore{hotpow}
  \end{algorithmic}
\end{listing}

A key difference to \nc{} is that the proof-of-work solutions in the quorum are
bound to the previous block and not to the state update of the proposed block
(see line~\ref{l:predecessor}).
This implements the separation of puzzle solutions from block proposals and
enables parallel puzzle solving (see Sect.~\ref{sec:intuition}).

\subsubsection{Proposing} \label{sssec:proto_propose}
Nodes assume leadership whenever possible. If so, the procedure
\Call{proposeIfLeader}{} (Listing~\ref{lst:propose}) obtains a state update
from the application, integrates it into a new valid block, and shares it with
the other nodes.

\begin{listing}[H]
  \caption{Block Proposals}
  \label{lst:propose}
  \begin{algorithmic}
    \algrestore{hotpow}
    \Procedure{proposeIfLeader}{\opRef{}}
      \If{$\exists$ valid \qsize{}-quorum $Q \subset$ blocks[$r$].votes \textbf{where} $Q$[0].$p$ = me}
        \State B.parent $\gets r$
        \State B.quorum $\gets Q$
        \State B.payload $\gets$ \Call{proposeUpdate}{blocks[\opRef{}].state}
        \State B.signature $\gets$ \Call{sign}{(B.parent, B.quorum, B.payload), secret}
        \State \Call{store}{B}
        \State \Schedule{send $\mid$ block B} \label{l:sendblock}
        \State \Return true
      \Else ~\Return false \EndIf
    \EndProcedure
    \algstore{hotpow}
  \end{algorithmic}
\end{listing}

\subsubsection{Commit} \label{sssec:proto_commit}
Proposals become final after the three-phase commit. Each subsequent block
carries a quorum that completes one phase, like in HotStuff (see
Sect.~\ref{sec:intuition}).
Consequently, the most recent application state can be retrieved from the local
block store as shown in Listing~\ref{lst:read_state}.

\begin{listing}[H]
  \caption{Reading Application State}
  \label{lst:read_state}
  \begin{algorithmic}
    \algrestore{hotpow}
    \Procedure{readState}{}
      \State \Return blocks[head].parent.parent.parent.state
    \EndProcedure
    \algstore{hotpow}
  \end{algorithmic}
\end{listing}

\subsubsection{Conflict Resolution} \label{sssec:proto_progress}
The commit becomes effective after three blocks, but we have to consider
conflicting block proposals at the uncommitted frontier.
For example, when more than $\qsize$~votes exist, the leader election is not
unique.
Moreover, a malicious leader can send different proposals without solving
additional proof-of-work puzzles.
Nodes resolve such conflicts based on the progress towards the \emph{next}
quorum.

\subsubsection{Block Preference} \label{sssec:proto_preference}
When learning of a new block or vote, nodes update their preferred chain
according to a modified version of Nakamoto's longest chain rule.
\theProtocol{} adapts it to include information on quorum progress
(Sect.~\ref{sssec:proto_progress}) and reject changes to already committed
state (Sect.~\ref{sssec:proto_commit}).
Procedure \Call{updateHead}{} (Listing~\ref{lst:update_head}) takes a candidate
block reference and updates the preferred chain if necessary.

\begin{listing}[H]
  \caption{Block Preference}
  \label{lst:update_head}
  \begin{algorithmic}
    \algrestore{hotpow}
    \Procedure{updateHead}{\opRef{}}
      \State H $\gets$ blocks[head]
      \State R $\gets$ blocks[\opRef{}]
      \State $d \gets \text{R.height} - \text{H.height}$
      \If{$d > 0 \vee (d = 0$ $\wedge$ $|\text{R.votes}| > |\text{H.votes}|)$}
        \While {$d > 0$} $($R, $d) \gets ($R.parent, $d - 1)$ \EndWhile
        \If {H.parent.parent.parent.block = R.parent.parent.parent.block}
        \State head $\gets \opRef$ \EndIf \label{l:detectfork}
      \EndIf
    \EndProcedure
    \algstore{hotpow}
  \end{algorithmic}
\end{listing}

\subsubsection{Main Program}
Listing~\ref{lst:hotpow} shows the set of event handlers that tie everything together and define a \theProtocol{}
node. The execution is initiated by scheduling the \event{init} event. The
listing shows how nodes assume leadership upon completing a suitable quorum
with an ATV of their own (line~\ref{l:leadershipa}), or votes received from
others, either directly (line~\ref{l:leadershipb}) or as part of a block
proposal (line~\ref{l:leadershipc}). In the last case, if more than $\qsize$
votes exist, it can happen that a node replaces the leader. It proposes a block
of its own by reusing votes contained in the received proposal. This is
possible because votes in \theProtocol{} reference the previous block and not
the current proposal. The possibility of reusing votes reduces wasted work
compared to orphans in \nc{}, a problem that has been studied
separately~\cite{sompolinsky2015SecureHighRate}.
It also provides robustness against leader failure (see Sect.~\ref{sssec:leader-failure}).

Line~\ref{l:onatv} handles ATVs.
If the node cannot lead a quorum, it broadcasts the vote.
The last missing part is how ATVs can be scheduled, which we discuss next.

\begin{listing}[H]
  \caption{The \theProtocol{} Protocol}
  \label{lst:hotpow}
  \begin{algorithmic}
    \algrestore{hotpow}
    \Event{init}
      \State me, secret $\gets$ \Call{generateKeyPair}{\,}
      \State head $\gets$ genesis \Comment{hard-coded magic value}
      \State blocks[genesis].state $\gets$ \initial
        \Comment{application's initial state}
      \State blocks[genesis].height $\gets$ 0
    \EndEvent

    \Event{atv $\mid$ \opSol{}} \label{l:onatv}
    \State \Call{collect}{head, me, \opSol{}}
      \If {\textbf{not} \Call{proposeIfLeader}{head}} \label{l:leadershipa}
        \State \Schedule{send $\mid$ vote (head, me, \opSol{})} \label{l:sendvote}
      \EndIf
    \EndEvent

    \Event{receive $\mid$ vote \vote{}} \Comment{sent by other node in line~\ref{l:sendvote}}
      \State \Call{collect}{\votenp{}}
      \State \Call{proposeIfLeader}{\opRef{}} \label{l:leadershipb}
    \EndEvent

    \Event{receive $\mid$ block B}  \Comment{sent by other node in line~\ref{l:sendblock} (Sect.~\ref{sssec:proto_propose})}
    \ForAll{$(\opId, \opSol)$ \textbf{in} B.quorum}
      \State \Call{collect}{B.parent, \opId{}, \opSol{}}
      \EndFor
      \State \Call{proposeIfLeader}{B.parent} \label{l:leadershipc}
      \If {\Call{validBlock}{B}} \Call{store}{B} \EndIf
    \EndEvent
    \algstore{hotpow}
  \end{algorithmic}
\end{listing}

\subsubsection{Work} \label{sssec:proto_work}

Agents can participate in the quorum finding process by computing ATVs on their
nodes. For completeness, Listing~\ref{lst:work} shows the trial-and-error
algorithm which schedules solutions suitable for votes ($\leq\vthres$).
Alternatively, agents can search ATVs with the help of other machines, possibly
in parallel and using specialized hardware. Figure~\ref{fig:appstack} reflects
this by splitting the lower layer in network and work.

\begin{listing}
  \caption{Puzzle Solving}
  \label{lst:work}
  \begin{algorithmic}
    \algrestore{hotpow}
    \Procedure{work}{}
      \State draw random number $n$
      \If {$\powhash(\text{head, me, }n) \leq \vthres$}
        \State \Schedule{atv $\mid$ n}
      \EndIf
      \State \Call{work}{}
    \EndProcedure
  \end{algorithmic}
\end{listing}

Figure~\ref{fig:timeline} in the appendix visualizes an execution of \theProtocol{} by correct nodes and compares it to \nc{}.

\subsection{Incentives} \label{ssec:incentives}

It is possible to motivate participation in \theProtocol{} by rewarding puzzle solutions. 
This requires some kind of virtual asset that (at least partly) fulfills the functions of
money~\cite[p.~1]{hicks1967CriticalEssays} and can be transferred to a vote's public key. Claiming the reward for $(\opRef,
\opId, \opSol)$ depends on the corresponding secret key.  %

\theProtocol{} could adopt Bobtails's constant reward per vote \cite{bissias2020BobtailImproved}.
Rewarding votes instead of blocks would ensure inclusiveness without compromising security (see Sect.~\ref{sec:intuition}).
Votes occur $\qsize$ times more frequently than blocks. \theProtocol's mining income would thus be less volatile than in \nc. This reduces the pressure to form mining pools. %

However, it is not trivial to establish if constant rewards are incentive compatible because the utility of the reward
\emph{outside} the system may affect the willingness to participate \emph{in} the system and
thereby make $\lambda$ endogenous~\cite{dimitri2017BitcoinMining,
prat2018EquilibriumModel}.  This implies that rewards must be treated
jointly with the assumptions preventing the failure modes
\ref{pow-network} and \ref{pow-adversary}. We are unaware of
protocol analyses that solve this problem convincingly.

On a more general note, designing protocols like economic mechanisms by
incentivizing desired behavior sounds attractive because there is some hope
that the assumption of honest nodes can be replaced by a somewhat weaker
assumption of rational agents~\cite{garay2013RationalProtocol,groce2012ByzantineAgreement}.
In this spirit, \citet{badertscher2018WhyDoes} present positive results for
Bitcoin in a discrete round execution model and under assumption of a constant
exchange rate.
However, many roadblocks remain. Agents' actions are not fully observable
(\eg{} information withholding) and preference orders are not fully knowable,
hence rationality is not precisely defined.  Side-payments (bribes), which
cannot be ruled out, pose an insurmountable challenge for mechanism
design~\cite{bonneau2016WhyBuy, judmayer2017MergedMining,
budish2018EconomicLimits}.  For distributed logs, which work inherently
sequential, this approach may even be thwarted by negative results
on the existence of unique equilibria in repeated
games~\cite{friedman1971NoncooperativeEquilibrium}.
For these reasons, we
skip the mechanism design aspects and limit our
contribution to transferring Byzantine consensus to proof-of-work scenarios. In
other words, \theProtocol{} supports incentives for inclusiveness, %
but its security intentionally does not rely on incentives.

\section{Evaluation} \label{sec:evaluation}

\newcommand{\simNNodes}{1000}
\newcommand{\simNBlocks}{500}
\newcommand{\simNRuns}{100}

We implement \theProtocol{} in OCaml and evaluate it in a network
of~\simNNodes{} nodes using a discrete event simulation.  We average
over~\simNRuns{} independent executions of the first~\simNBlocks{} blocks. All
results are reproducible with the code provided online.\repofootnote{}

\begin{figure}
  \centering
  \resizebox{\columnwidth}{!}{
    \begin{tikzpicture}[x=2.60cm, y=0.7cm, >=stealth]
      \node at (0.25, -0.5) {priority queue};
      \node at (0.25, -1.0) {\scriptsize $t_i \leq t_{i+1}$};
      \draw (-0.5,0) rectangle (1,1);
      \draw (-0.5,1) rectangle (1,2);
      \draw (-0.5,2) rectangle (1,3);
      \draw (-0.5,3) rectangle (1,4);
      \draw [dotted] (-0.5,4) -- (-0.5,5);
      \draw [dotted] (1,4) -- (1,5);
      \draw (-0.5,5) rectangle (1,6);
      \draw (-0.5,6) rectangle (1,7);
      \draw [dotted] (-0.5,7) -- (-0.5,8);
      \draw [dotted] (1,7) -- (1,8);
      \node [anchor=west] at (-0.5, 0.5) {$t_0$};
      \node [anchor=west] at (-0.5, 1.5) {$t_1$};
      \node [anchor=west] at (-0.5, 2.5) {$t_2$};
      \node [anchor=west] at (-0.5, 3.5) {$t_3$};
      \node [anchor=west] at (-0.5, 5.5) {$t_{n+1}$};
      \node [anchor=west] at (-0.5, 6.5) {$t_{n+2}$};
      \node [anchor=west] at (-0.5 + 0.15, 0.5) {\event{atv}};
      \node [anchor=west] at (-0.5 + 0.15, 1.5) {\event{broadcast $\mid m$}};
      \node [anchor=west] at (-0.5 + 0.15, 2.5) {\event{deliver $\mid m, n-1$}};
      \node [anchor=west] at (-0.5 + 0.15, 3.5) {\event{deliver $\mid m, n-2$}};
      \node [anchor=west] at (-0.5 + 0.33, 5.5) {\event{deliver $\mid m, 0$}};
      \node [anchor=west] at (-0.5 + 0.33, 6.5) {\event{atv}};

      \draw (1.2, 0) rectangle (2.8, 8);
      \node at (2, -0.5) {simulation};

      \draw [fill=white,dashed] (3.0666, 0.2) rectangle (4.0666, 3.2);
      \draw [fill=white] (3.0333, 0.1) rectangle (4.0333, 3.1);
      \draw [fill=white] (3, 0) rectangle (4, 3.0);
      \node at (3.5, -0.5) {honest nodes};
      \node [anchor=west] at (3.07, 0.5) {\event{atv}};
      \node [anchor=west] at (3.07, 2.5) {\event{receive $\mid m$}};
      \node [anchor=west] at (3.07, 1.5) {\event{send $\mid m$}};

      \draw (3, 5) rectangle (4.00, 8);
      \node at (3.45, 4.5) {attacker node};
      \node [anchor=west] at (3.07, 6.5) {\event{atv}};
      \node [anchor=west] at (3.07, 5.5) {\event{receive $\mid m$}};
      \node [anchor=west] at (3.07, 7.5) {\event{send $\mid m'$}};

      \path [->] (1,0.5) edge node[above] {random assignment} (3.09, 0.5);
      \path [->] (3.07,1.5) edge node[above] {} (1, 1.5);
      \path [->] (1,2.5) edge (3.09, 2.5);
      \path [->] (1,5.5) edge (3.09, 5.5);
      \path [->] (1,6.5) edge node[above] {random assignment} (3.09, 6.5);

      \path (1,1.75) edge (1.4, 1.75);
      \path (1.4,1.75) edge (1.4, 5.75);
      \path [<-] (1,2.75) edge (1.4, 2.75);
      \path [<-] (1,3.75) edge (1.4, 3.75);
      \path [<-] (1,5.75) edge (1.4, 5.75);
      \node [anchor=west, align=left] at (1.43, 4) {model latency,\\churn, and failure};

      \draw[->] (1,-2) -- +(1,0) node [anchor=west] {$\quad$schedule event};

    \end{tikzpicture}
  }
  \caption{Schematic overview of the discrete event simulation.}
  \label{fig:simulator}
\end{figure}
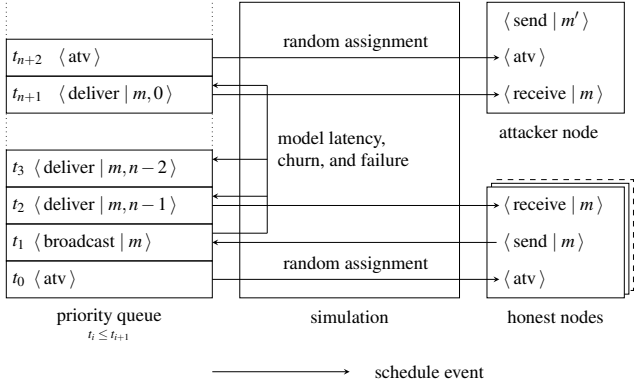

The simulation maintains state for all simulated nodes separately.  Events are
stored in a priority queue, with keys representing points in time.  Events are
scheduled by inserting them into the queue. There are three types of simulation
events: \event{atv}, \event{broadcast} and \event{deliver}. The simulation's
main loop takes the first event from the queue and handles it by interacting
with the nodes in the following way (also see Fig.~\ref{fig:simulator}).

\paragraph{Proof-of-Work} When taking an \event{atv} event from the queue, the
simulation randomly and independently assigns an ATV to a node. The simulation
executes the assignment by invoking the \event{atv} event handler on the
receiving node.  Then, it schedules the next ATV with a random, exponentially
distributed time delta.  This simulates a proof-of-work process according to
Def.~\ref{def:process}.
The simulation does not perform actual work by setting the vote threshold
\vthres{} to the maximum; meaning puzzles are trivial to solve.

\paragraph{Broadcast} Nodes invoke the broadcast logic by scheduling local
\event{send~$\mid$~.} events. The simulation translates them to global
\event{broadcast~$\mid$~.} events. For each broadcast event, the simulation
schedules \event{deliver~$\mid$~.} events for each node except the sender.
During this step, the simulation injects latency and simulates churn and leader
failure. Delivery events are handled by invoking the \event{receive~$\mid$~.}
handler on the receiving node.

\subsection{Robustness} \label{ssec:robustness}

We evaluate the robustness in terms of latency, churn, and leader failure. In all simulation runs we check for inconsistent
committed state, which did not occur.

\begin{figure}
  \resizebox{\columnwidth}{!}{
\begin{tikzpicture}[x=1pt,y=1pt]
\definecolor{fillColor}{RGB}{255,255,255}
\path[use as bounding box,fill=fillColor,fill opacity=0.00] (0,0) rectangle (267.40,137.31);
\begin{scope}
\path[clip] (  0.00,  0.00) rectangle (267.40,137.31);
\definecolor{drawColor}{RGB}{0,0,0}

\path[draw=drawColor,line width= 0.4pt,line join=round,line cap=round] ( 45.60, 35.07) -- ( 45.60,131.66);

\path[draw=drawColor,line width= 0.4pt,line join=round,line cap=round] ( 45.60, 35.07) -- ( 39.60, 35.07);

\path[draw=drawColor,line width= 0.4pt,line join=round,line cap=round] ( 45.60, 59.22) -- ( 39.60, 59.22);

\path[draw=drawColor,line width= 0.4pt,line join=round,line cap=round] ( 45.60, 83.36) -- ( 39.60, 83.36);

\path[draw=drawColor,line width= 0.4pt,line join=round,line cap=round] ( 45.60,107.51) -- ( 39.60,107.51);

\path[draw=drawColor,line width= 0.4pt,line join=round,line cap=round] ( 45.60,131.66) -- ( 39.60,131.66);

\node[text=drawColor,anchor=base east,inner sep=0pt, outer sep=0pt, scale=  1.00] at ( 33.60, 31.62) {1.00};

\node[text=drawColor,anchor=base east,inner sep=0pt, outer sep=0pt, scale=  1.00] at ( 33.60, 55.77) {1.05};

\node[text=drawColor,anchor=base east,inner sep=0pt, outer sep=0pt, scale=  1.00] at ( 33.60, 79.92) {1.10};

\node[text=drawColor,anchor=base east,inner sep=0pt, outer sep=0pt, scale=  1.00] at ( 33.60,104.07) {1.15};

\node[text=drawColor,anchor=base east,inner sep=0pt, outer sep=0pt, scale=  1.00] at ( 33.60,128.22) {1.20};

\path[draw=drawColor,line width= 0.4pt,line join=round,line cap=round] ( 45.60, 31.20) --
	(266.20, 31.20) --
	(266.20,136.11) --
	( 45.60,136.11) --
	( 45.60, 31.20);
\end{scope}
\begin{scope}
\path[clip] (  0.00,  0.00) rectangle (267.40,137.31);
\definecolor{drawColor}{RGB}{0,0,0}

\node[text=drawColor,rotate= 90.00,anchor=base,inner sep=0pt, outer sep=0pt, scale=  1.00] at (  7.20, 83.66) {time to commit};

\node[text=drawColor,anchor=base,inner sep=0pt, outer sep=0pt, scale=  1.00] at (155.90,  3.60) {latency};
\end{scope}
\begin{scope}
\path[clip] (  0.00,  0.00) rectangle (267.40,137.31);
\definecolor{drawColor}{RGB}{0,0,0}

\path[draw=drawColor,line width= 0.4pt,line join=round,line cap=round] ( 45.60, 31.20) -- (266.20, 31.20);

\path[draw=drawColor,line width= 0.4pt,line join=round,line cap=round] ( 53.77, 31.20) -- ( 53.77, 25.20);

\path[draw=drawColor,line width= 0.4pt,line join=round,line cap=round] (121.86, 31.20) -- (121.86, 25.20);

\path[draw=drawColor,line width= 0.4pt,line join=round,line cap=round] (189.94, 31.20) -- (189.94, 25.20);

\path[draw=drawColor,line width= 0.4pt,line join=round,line cap=round] (258.03, 31.20) -- (258.03, 25.20);

\node[text=drawColor,anchor=base,inner sep=0pt, outer sep=0pt, scale=  1.00] at ( 53.77,  9.60) {$10^{- 4 }$};

\node[text=drawColor,anchor=base,inner sep=0pt, outer sep=0pt, scale=  1.00] at (121.86,  9.60) {$10^{- 3 }$};

\node[text=drawColor,anchor=base,inner sep=0pt, outer sep=0pt, scale=  1.00] at (189.94,  9.60) {$10^{- 2 }$};

\node[text=drawColor,anchor=base,inner sep=0pt, outer sep=0pt, scale=  1.00] at (258.03,  9.60) {$10^{- 1 }$};

\path[draw=drawColor,line width= 0.4pt,line join=round,line cap=round] ( 45.60, 31.20) -- (266.20, 31.20);

\path[draw=drawColor,line width= 0.4pt,line join=round,line cap=round] ( 53.77, 31.20) -- ( 53.77, 25.20);

\path[draw=drawColor,line width= 0.4pt,line join=round,line cap=round] (258.03, 31.20) -- (258.03, 25.20);

\node[text=drawColor,anchor=base,inner sep=0pt, outer sep=0pt, scale=  1.00] at ( 53.77,  1.20) {\scriptsize 60ms};

\node[text=drawColor,anchor=base,inner sep=0pt, outer sep=0pt, scale=  1.00] at (258.03,  1.20) {\scriptsize 1'};
\end{scope}
\begin{scope}
\path[clip] ( 45.60, 31.20) rectangle (266.20,136.11);
\definecolor{drawColor}{RGB}{0,0,0}

\path[draw=drawColor,line width= 0.4pt,line join=round,line cap=round] ( 53.77, 37.28) --
	( 87.79, 36.44) --
	(121.86, 35.90) --
	(155.88, 39.19) --
	(189.94, 44.19) --
	(223.96, 59.60) --
	(258.03,111.61);

\path[draw=drawColor,line width= 0.4pt,dash pattern=on 4pt off 4pt ,line join=round,line cap=round] ( 53.77, 35.24) --
	( 87.79, 35.09) --
	(121.86, 36.17) --
	(155.88, 36.74) --
	(189.94, 46.19) --
	(223.96, 64.89) --
	(258.03,127.35);

\path[draw=drawColor,line width= 0.4pt,dash pattern=on 1pt off 3pt ,line join=round,line cap=round] ( 53.77, 35.60) --
	( 87.79, 35.47) --
	(121.86, 35.96) --
	(155.88, 38.24) --
	(189.94, 44.88) --
	(223.96, 65.01) --
	(258.03,131.65);

\path[draw=drawColor,line width= 0.4pt,dash pattern=on 1pt off 3pt on 4pt off 3pt ,line join=round,line cap=round] ( 53.77, 35.09) --
	( 87.79, 35.27) --
	(121.86, 35.80) --
	(155.88, 37.93) --
	(189.94, 45.00) --
	(223.96, 65.73) --
	(258.03,132.23);

\path[draw=drawColor,line width= 0.4pt,line join=round,line cap=round] (205.05, 42.81) --
	( 68.88, 42.81);

\path[draw=drawColor,line width= 0.4pt,dash pattern=on 4pt off 4pt ,line join=round,line cap=round] (205.05, 44.31) --
	( 68.88, 44.31);

\path[draw=drawColor,line width= 0.4pt,dash pattern=on 1pt off 3pt ,line join=round,line cap=round] (205.05, 43.33) --
	( 68.88, 43.33);

\path[draw=drawColor,line width= 0.4pt,dash pattern=on 1pt off 3pt on 4pt off 3pt ,line join=round,line cap=round] (205.05, 43.74) --
	( 68.88, 43.74);

\path[draw=drawColor,line width= 0.4pt,line join=round,line cap=round] ( 45.60,136.11) rectangle (101.10, 64.11);

\path[draw=drawColor,line width= 0.4pt,line join=round,line cap=round] ( 54.60,112.11) -- ( 72.60,112.11);

\path[draw=drawColor,line width= 0.4pt,dash pattern=on 4pt off 4pt ,line join=round,line cap=round] ( 54.60,100.11) -- ( 72.60,100.11);

\path[draw=drawColor,line width= 0.4pt,dash pattern=on 1pt off 3pt ,line join=round,line cap=round] ( 54.60, 88.11) -- ( 72.60, 88.11);

\path[draw=drawColor,line width= 0.4pt,dash pattern=on 1pt off 3pt on 4pt off 3pt ,line join=round,line cap=round] ( 54.60, 76.11) -- ( 72.60, 76.11);

\node[text=drawColor,anchor=base,inner sep=0pt, outer sep=0pt, scale=  1.00] at ( 73.35,124.11) {$k$};

\node[text=drawColor,anchor=base west,inner sep=0pt, outer sep=0pt, scale=  1.00] at ( 81.60,108.67) {2};

\node[text=drawColor,anchor=base west,inner sep=0pt, outer sep=0pt, scale=  1.00] at ( 81.60, 96.67) {8};

\node[text=drawColor,anchor=base west,inner sep=0pt, outer sep=0pt, scale=  1.00] at ( 81.60, 84.67) {32};

\node[text=drawColor,anchor=base west,inner sep=0pt, outer sep=0pt, scale=  1.00] at ( 81.60, 72.67) {128};
\end{scope}
\end{tikzpicture}}

  \caption{The effect of latency on the time to commit. The latency is stated
    relative to the optimistic quorum time (in small print for a quorum time of
    10'). The horizontal lines show a realistic scenario ($10$s for
    blocks, $100$ms for votes).
  }
  \label{fig:latency}
\end{figure}
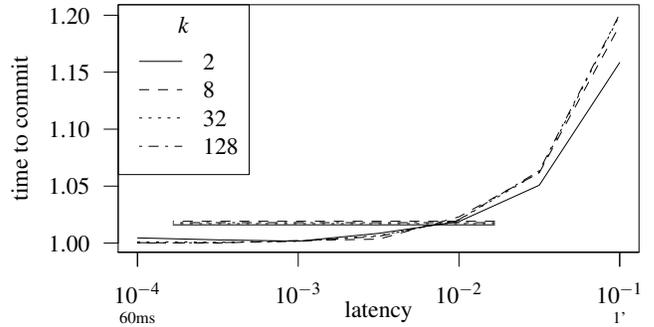

\subsubsection{Latency}\label{sssec:latency} We model the effect of latency by
injecting a random time delay between broadcast send and message delivery. We
draw delays from an exponential distribution with fixed expectation,
independently for each node and delivery.  Latency causes temporal state
inconsistencies.  In these periods, nodes spend their ATVs on extending
superseded blocks, or even produce temporal forks.
We observe that largely independent of the quorum size $\qsize$, expected
latencies below 1\,\% of the expected block time (Bitcoin: 6 seconds) have
marginal impact, while latencies in the order of 10\,\% of the expected block
time (Bitcoin: 60 seconds) delay the commit by about 20\,\%.
Figure~\ref{fig:latency} visualizes these results.

Empirical measurements~\cite{decker2013InformationPropagation,
croman2016ScalingDecentralized,gervais2016SecurityPerformance} suggest that the propagation
time of Bitcoin blocks ($\approx 500$ KB) is about 9 seconds on the Internet.
If we take this as an upper bound, we can argue that \theProtocol{} tolerates
practical latencies.  Moreover, most of \theProtocol{}'s messages are votes.
They are multiple orders of magnitude smaller ($72$ B; see
Sect.~\ref{ssec:overhead}), fit into a single packet, and are much easier to
verify than Bitcoin blocks.
Results of a simulation with different latencies for blocks ($10$s) and votes ($100$ms) suggest that \theProtocol{} can run at Internet scale with lower expected block time than 10 minutes.

\begin{figure}
  \resizebox{\columnwidth}{!}{
\begin{tikzpicture}[x=1pt,y=1pt]
\definecolor{fillColor}{RGB}{255,255,255}
\path[use as bounding box,fill=fillColor,fill opacity=0.00] (0,0) rectangle (267.40,137.31);
\begin{scope}
\path[clip] (  0.00,  0.00) rectangle (267.40,137.31);
\definecolor{drawColor}{RGB}{0,0,0}

\path[draw=drawColor,line width= 0.4pt,line join=round,line cap=round] ( 45.60, 35.41) -- ( 45.60,126.79);

\path[draw=drawColor,line width= 0.4pt,line join=round,line cap=round] ( 45.60, 35.41) -- ( 39.60, 35.41);

\path[draw=drawColor,line width= 0.4pt,line join=round,line cap=round] ( 45.60, 53.69) -- ( 39.60, 53.69);

\path[draw=drawColor,line width= 0.4pt,line join=round,line cap=round] ( 45.60, 71.96) -- ( 39.60, 71.96);

\path[draw=drawColor,line width= 0.4pt,line join=round,line cap=round] ( 45.60, 90.24) -- ( 39.60, 90.24);

\path[draw=drawColor,line width= 0.4pt,line join=round,line cap=round] ( 45.60,108.52) -- ( 39.60,108.52);

\path[draw=drawColor,line width= 0.4pt,line join=round,line cap=round] ( 45.60,126.79) -- ( 39.60,126.79);

\node[text=drawColor,anchor=base east,inner sep=0pt, outer sep=0pt, scale=  1.00] at ( 33.60, 31.97) {1.0};

\node[text=drawColor,anchor=base east,inner sep=0pt, outer sep=0pt, scale=  1.00] at ( 33.60, 50.24) {1.2};

\node[text=drawColor,anchor=base east,inner sep=0pt, outer sep=0pt, scale=  1.00] at ( 33.60, 68.52) {1.4};

\node[text=drawColor,anchor=base east,inner sep=0pt, outer sep=0pt, scale=  1.00] at ( 33.60, 86.80) {1.6};

\node[text=drawColor,anchor=base east,inner sep=0pt, outer sep=0pt, scale=  1.00] at ( 33.60,105.07) {1.8};

\node[text=drawColor,anchor=base east,inner sep=0pt, outer sep=0pt, scale=  1.00] at ( 33.60,123.35) {2.0};

\path[draw=drawColor,line width= 0.4pt,line join=round,line cap=round] ( 45.60, 31.20) --
	(266.20, 31.20) --
	(266.20,136.11) --
	( 45.60,136.11) --
	( 45.60, 31.20);
\end{scope}
\begin{scope}
\path[clip] (  0.00,  0.00) rectangle (267.40,137.31);
\definecolor{drawColor}{RGB}{0,0,0}

\node[text=drawColor,rotate= 90.00,anchor=base,inner sep=0pt, outer sep=0pt, scale=  1.00] at (  7.20, 83.66) {time to commit};
\end{scope}
\begin{scope}
\path[clip] (  0.00,  0.00) rectangle (267.40,137.31);
\definecolor{drawColor}{RGB}{0,0,0}

\path[draw=drawColor,line width= 0.4pt,line join=round,line cap=round] ( 53.77, 31.20) -- (258.03, 31.20);

\path[draw=drawColor,line width= 0.4pt,line join=round,line cap=round] ( 53.77, 31.20) -- ( 53.77, 25.20);

\path[draw=drawColor,line width= 0.4pt,line join=round,line cap=round] ( 94.62, 31.20) -- ( 94.62, 25.20);

\path[draw=drawColor,line width= 0.4pt,line join=round,line cap=round] (135.47, 31.20) -- (135.47, 25.20);

\path[draw=drawColor,line width= 0.4pt,line join=round,line cap=round] (176.33, 31.20) -- (176.33, 25.20);

\path[draw=drawColor,line width= 0.4pt,line join=round,line cap=round] (217.18, 31.20) -- (217.18, 25.20);

\path[draw=drawColor,line width= 0.4pt,line join=round,line cap=round] (258.03, 31.20) -- (258.03, 25.20);

\node[text=drawColor,anchor=base,inner sep=0pt, outer sep=0pt, scale=  1.00] at ( 53.77,  9.60) {0};

\node[text=drawColor,anchor=base,inner sep=0pt, outer sep=0pt, scale=  1.00] at ( 94.62,  9.60) {0.1};

\node[text=drawColor,anchor=base,inner sep=0pt, outer sep=0pt, scale=  1.00] at (217.18,  9.60) {0.4};

\node[text=drawColor,anchor=base,inner sep=0pt, outer sep=0pt, scale=  1.00] at (258.03,  9.60) {0.5};
\end{scope}
\begin{scope}
\path[clip] (  0.00,  0.00) rectangle (267.40,137.31);
\definecolor{drawColor}{RGB}{0,0,0}

\node[text=drawColor,anchor=base,inner sep=0pt, outer sep=0pt, scale=  1.00] at (155.90,  3.60) {churn ratio};
\end{scope}
\begin{scope}
\path[clip] ( 45.60, 31.20) rectangle (266.20,136.11);
\definecolor{drawColor}{RGB}{0,0,0}

\path[draw=drawColor,line width= 0.4pt,line join=round,line cap=round] ( 53.77, 35.39) --
	( 57.86, 36.34) --
	( 74.20, 40.83) --
	( 94.62, 45.91) --
	(135.47, 59.68) --
	(176.33, 76.08) --
	(217.18, 98.38) --
	(258.03,132.23);

\path[draw=drawColor,line width= 0.4pt,dash pattern=on 4pt off 4pt ,line join=round,line cap=round] ( 53.77, 35.52) --
	( 57.86, 36.22) --
	( 74.20, 40.16) --
	( 94.62, 45.48) --
	(135.47, 58.71) --
	(176.33, 75.18) --
	(217.18, 97.48) --
	(258.03,129.05);

\path[draw=drawColor,line width= 0.4pt,dash pattern=on 1pt off 3pt ,line join=round,line cap=round] ( 53.77, 35.40) --
	( 57.86, 36.38) --
	( 74.20, 40.38) --
	( 94.62, 45.59) --
	(135.47, 58.30) --
	(176.33, 74.70) --
	(217.18, 96.82) --
	(258.03,127.26);

\path[draw=drawColor,line width= 0.4pt,dash pattern=on 1pt off 3pt on 4pt off 3pt ,line join=round,line cap=round] ( 53.77, 35.39) --
	( 57.86, 36.34) --
	( 74.20, 40.20) --
	( 94.62, 45.60) --
	(135.47, 58.35) --
	(176.33, 74.65) --
	(217.18, 96.37) --
	(258.03,127.04);

\path[draw=drawColor,line width= 0.4pt,line join=round,line cap=round] ( 45.60,136.11) rectangle (101.10, 64.11);

\path[draw=drawColor,line width= 0.4pt,line join=round,line cap=round] ( 54.60,112.11) -- ( 72.60,112.11);

\path[draw=drawColor,line width= 0.4pt,dash pattern=on 4pt off 4pt ,line join=round,line cap=round] ( 54.60,100.11) -- ( 72.60,100.11);

\path[draw=drawColor,line width= 0.4pt,dash pattern=on 1pt off 3pt ,line join=round,line cap=round] ( 54.60, 88.11) -- ( 72.60, 88.11);

\path[draw=drawColor,line width= 0.4pt,dash pattern=on 1pt off 3pt on 4pt off 3pt ,line join=round,line cap=round] ( 54.60, 76.11) -- ( 72.60, 76.11);

\node[text=drawColor,anchor=base,inner sep=0pt, outer sep=0pt, scale=  1.00] at ( 73.35,124.11) {$k$};

\node[text=drawColor,anchor=base west,inner sep=0pt, outer sep=0pt, scale=  1.00] at ( 81.60,108.67) {2};

\node[text=drawColor,anchor=base west,inner sep=0pt, outer sep=0pt, scale=  1.00] at ( 81.60, 96.67) {8};

\node[text=drawColor,anchor=base west,inner sep=0pt, outer sep=0pt, scale=  1.00] at ( 81.60, 84.67) {32};

\node[text=drawColor,anchor=base west,inner sep=0pt, outer sep=0pt, scale=  1.00] at ( 81.60, 72.67) {128};
\end{scope}
\end{tikzpicture}}

  \caption{The effect of churn on the time to commit.}
  \label{fig:churn}
\end{figure}
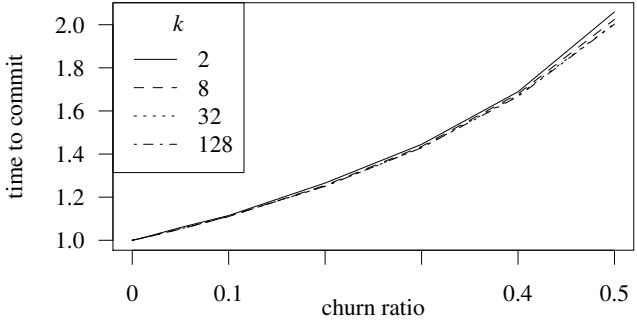

\subsubsection{Churn}
We simulate churn by muting a fraction (churn ratio) of random nodes for 10
times the expected block time.  Muted nodes can receive ATVs but do not send or
receive messages.  Accordingly, the ATVs assigned to muted nodes represent lost
work. We expect that the time to commit is inversely proportional to the churn
ratio: if 50\,\% of the nodes are muted, the time to commit is twice as long,
independent of the quorum size.  Figure~\ref{fig:churn} supports this claim.

\subsubsection{Leader Failure} \label{sssec:leader-failure}
Leaders may fail to propose blocks. We model such failures by dropping block
proposals randomly with constant probability (leader failure rate).

In \nc{}, lost proposals imply a full block worth of wasted work.
\theProtocol{} can reuse votes for different proposals.  Honest nodes reveal at
most one new vote with their proposal. Accordingly, a lost proposal wastes at
most the work of one vote. Therefore, with increasing quorum size the
robustness to leader failure should improve.  The results in
Figure~\ref{fig:failure_real} (with realistic 10s/100ms latency) and
Figure~\ref{fig:failure} (without latency to isolate effects) support this claim.  For perspective,
the right end of the graph simulates a situation where an attacker can
monitor all nodes' network traffic and disconnect nodes at discretion with
50\,\% success probability. Still, for large quorum sizes the time to commit is
not longer than under the extreme latencies discussed in
Section~\ref{sssec:latency}.

\begin{figure}
  \resizebox{\columnwidth}{!}{
\begin{tikzpicture}[x=1pt,y=1pt]
\definecolor{fillColor}{RGB}{255,255,255}
\path[use as bounding box,fill=fillColor,fill opacity=0.00] (0,0) rectangle (267.40,137.31);
\begin{scope}
\path[clip] (  0.00,  0.00) rectangle (267.40,137.31);
\definecolor{drawColor}{RGB}{0,0,0}

\path[draw=drawColor,line width= 0.4pt,line join=round,line cap=round] ( 45.60, 31.98) -- ( 45.60,128.54);

\path[draw=drawColor,line width= 0.4pt,line join=round,line cap=round] ( 45.60, 31.98) -- ( 39.60, 31.98);

\path[draw=drawColor,line width= 0.4pt,line join=round,line cap=round] ( 45.60, 51.29) -- ( 39.60, 51.29);

\path[draw=drawColor,line width= 0.4pt,line join=round,line cap=round] ( 45.60, 70.60) -- ( 39.60, 70.60);

\path[draw=drawColor,line width= 0.4pt,line join=round,line cap=round] ( 45.60, 89.91) -- ( 39.60, 89.91);

\path[draw=drawColor,line width= 0.4pt,line join=round,line cap=round] ( 45.60,109.22) -- ( 39.60,109.22);

\path[draw=drawColor,line width= 0.4pt,line join=round,line cap=round] ( 45.60,128.54) -- ( 39.60,128.54);

\node[text=drawColor,anchor=base east,inner sep=0pt, outer sep=0pt, scale=  1.00] at ( 33.60, 28.53) {1.0};

\node[text=drawColor,anchor=base east,inner sep=0pt, outer sep=0pt, scale=  1.00] at ( 33.60, 47.85) {1.1};

\node[text=drawColor,anchor=base east,inner sep=0pt, outer sep=0pt, scale=  1.00] at ( 33.60, 67.16) {1.2};

\node[text=drawColor,anchor=base east,inner sep=0pt, outer sep=0pt, scale=  1.00] at ( 33.60, 86.47) {1.3};

\node[text=drawColor,anchor=base east,inner sep=0pt, outer sep=0pt, scale=  1.00] at ( 33.60,105.78) {1.4};

\node[text=drawColor,anchor=base east,inner sep=0pt, outer sep=0pt, scale=  1.00] at ( 33.60,125.09) {1.5};

\path[draw=drawColor,line width= 0.4pt,line join=round,line cap=round] ( 45.60, 31.20) --
	(266.20, 31.20) --
	(266.20,136.11) --
	( 45.60,136.11) --
	( 45.60, 31.20);
\end{scope}
\begin{scope}
\path[clip] (  0.00,  0.00) rectangle (267.40,137.31);
\definecolor{drawColor}{RGB}{0,0,0}

\node[text=drawColor,rotate= 90.00,anchor=base,inner sep=0pt, outer sep=0pt, scale=  1.00] at (  7.20, 83.66) {time to commit};
\end{scope}
\begin{scope}
\path[clip] (  0.00,  0.00) rectangle (267.40,137.31);
\definecolor{drawColor}{RGB}{0,0,0}

\path[draw=drawColor,line width= 0.4pt,line join=round,line cap=round] ( 53.77, 31.20) -- (258.03, 31.20);

\path[draw=drawColor,line width= 0.4pt,line join=round,line cap=round] ( 53.77, 31.20) -- ( 53.77, 25.20);

\path[draw=drawColor,line width= 0.4pt,line join=round,line cap=round] ( 94.62, 31.20) -- ( 94.62, 25.20);

\path[draw=drawColor,line width= 0.4pt,line join=round,line cap=round] (135.47, 31.20) -- (135.47, 25.20);

\path[draw=drawColor,line width= 0.4pt,line join=round,line cap=round] (176.33, 31.20) -- (176.33, 25.20);

\path[draw=drawColor,line width= 0.4pt,line join=round,line cap=round] (217.18, 31.20) -- (217.18, 25.20);

\path[draw=drawColor,line width= 0.4pt,line join=round,line cap=round] (258.03, 31.20) -- (258.03, 25.20);

\node[text=drawColor,anchor=base,inner sep=0pt, outer sep=0pt, scale=  1.00] at ( 53.77,  9.60) {0};

\node[text=drawColor,anchor=base,inner sep=0pt, outer sep=0pt, scale=  1.00] at ( 94.62,  9.60) {0.1};

\node[text=drawColor,anchor=base,inner sep=0pt, outer sep=0pt, scale=  1.00] at (217.18,  9.60) {0.4};

\node[text=drawColor,anchor=base,inner sep=0pt, outer sep=0pt, scale=  1.00] at (258.03,  9.60) {0.5};
\end{scope}
\begin{scope}
\path[clip] (  0.00,  0.00) rectangle (267.40,137.31);
\definecolor{drawColor}{RGB}{0,0,0}

\node[text=drawColor,anchor=base,inner sep=0pt, outer sep=0pt, scale=  1.00] at (155.90,  3.60) {leader failure rate};
\end{scope}
\begin{scope}
\path[clip] ( 45.60, 31.20) rectangle (266.20,136.11);
\definecolor{drawColor}{RGB}{0,0,0}

\path[draw=drawColor,line width= 0.4pt,line join=round,line cap=round] ( 53.77, 35.95) --
	( 57.86, 36.56) --
	( 74.20, 40.33) --
	( 94.62, 45.69) --
	(135.47, 59.04) --
	(176.33, 77.32) --
	(217.18,101.15) --
	(258.03,132.23);

\path[draw=drawColor,line width= 0.4pt,dash pattern=on 4pt off 4pt ,line join=round,line cap=round] ( 53.77, 35.83) --
	( 57.86, 35.75) --
	( 74.20, 36.68) --
	( 94.62, 37.85) --
	(135.47, 41.41) --
	(176.33, 45.49) --
	(217.18, 51.31) --
	(258.03, 60.07);

\path[draw=drawColor,line width= 0.4pt,dash pattern=on 1pt off 3pt ,line join=round,line cap=round] ( 53.77, 35.09) --
	( 57.86, 35.46) --
	( 74.20, 35.76) --
	( 94.62, 36.06) --
	(135.47, 37.12) --
	(176.33, 38.05) --
	(217.18, 39.49) --
	(258.03, 41.63);

\path[draw=drawColor,line width= 0.4pt,dash pattern=on 1pt off 3pt on 4pt off 3pt ,line join=round,line cap=round] ( 53.77, 35.63) --
	( 57.86, 35.53) --
	( 74.20, 35.41) --
	( 94.62, 35.58) --
	(135.47, 35.92) --
	(176.33, 36.23) --
	(217.18, 36.60) --
	(258.03, 37.03);

\path[draw=drawColor,line width= 0.4pt,line join=round,line cap=round] ( 45.60,136.11) rectangle (101.10, 64.11);

\path[draw=drawColor,line width= 0.4pt,line join=round,line cap=round] ( 54.60,112.11) -- ( 72.60,112.11);

\path[draw=drawColor,line width= 0.4pt,dash pattern=on 4pt off 4pt ,line join=round,line cap=round] ( 54.60,100.11) -- ( 72.60,100.11);

\path[draw=drawColor,line width= 0.4pt,dash pattern=on 1pt off 3pt ,line join=round,line cap=round] ( 54.60, 88.11) -- ( 72.60, 88.11);

\path[draw=drawColor,line width= 0.4pt,dash pattern=on 1pt off 3pt on 4pt off 3pt ,line join=round,line cap=round] ( 54.60, 76.11) -- ( 72.60, 76.11);

\node[text=drawColor,anchor=base,inner sep=0pt, outer sep=0pt, scale=  1.00] at ( 73.35,124.11) {$k$};

\node[text=drawColor,anchor=base west,inner sep=0pt, outer sep=0pt, scale=  1.00] at ( 81.60,108.67) {2};

\node[text=drawColor,anchor=base west,inner sep=0pt, outer sep=0pt, scale=  1.00] at ( 81.60, 96.67) {8};

\node[text=drawColor,anchor=base west,inner sep=0pt, outer sep=0pt, scale=  1.00] at ( 81.60, 84.67) {32};

\node[text=drawColor,anchor=base west,inner sep=0pt, outer sep=0pt, scale=  1.00] at ( 81.60, 72.67) {128};
\end{scope}
\end{tikzpicture}}

  \caption{The effect of leader failure on the time to commit.}
  \label{fig:failure_real}
\end{figure}
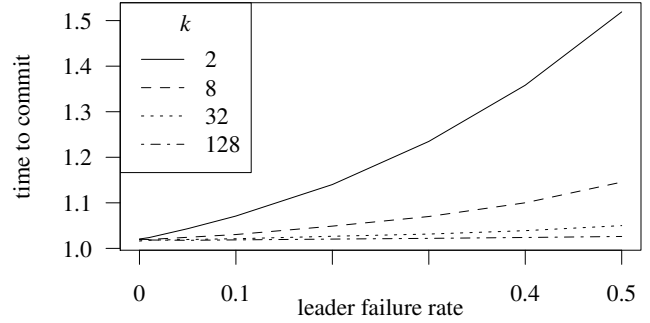

\medskip

The robustness against churn and leader failure emerges from \theProtocol{}'s novel approach to form short-lived committees from ephemeral identities. This maintains liveness even under the threat of powerful network-level attacks. We move on to the discussion of attacks on the protocol layer.

\subsection{Security} \label{ssec:security}

The security evaluation draws on the framework by~\citet{zhang2019LayCommon}. 
It distinguishes the security aspects proof-of-work blockchains should fulfill: chain quality, incentive compatibility, subversion gain, and censorship susceptibility.

The authors suggest Markov Decision Processes (MDP) as method and apply it to several variants of \nc{}. However, state explosion prevented them from modeling Bobtail,%
\footnote{\citet{zhang2019LayCommon} were aware of Bobtail and chose not to model it. This is confirmed in private communication with the authors of Bobtail~\cite{bissias2020BobtailImproved}.} 
because it ranks proof-of-work solutions by magnitude. Since \theProtocol{} adopts this ranking for the leader election (Sect.~\ref{sssec:proto_leader}), it does not seem readily amenable to MDPs, either. 
We thus resort to informal reasoning and simulation. %

Following the convention in the literature, we assume two agents. Let
$\lambda$ be the total compute power. The attacker has $\alpha\cdot\lambda$ compute power, the
honest agent controls the rest. The honest agent operates correct nodes, while the attacker 
operates a single node that may deviate from the protocol specification.

\subsubsection{Subversion Gain} 
The canonical example for subversion gain in cryptocurrencies is double spending: 
the attacker wants at least one of the honest nodes (the merchant) to act on
inconsistent state. \theProtocol{} supports commits, hence we neither need to
consider the possibility of history rewriting nor the double
spending of \emph{un}committed transactions.%
\footnote{Sound applications on a system with finality wait until the commit. \theProtocol{} can be parametrized to acceptable commit times for
economic exchanges between humans. (High-frequency trading needs other
architectures.)}
\nc{} suffers from these problems~\cite{karame2012DoublespendingFast, heilman2015EclipseAttacks,
gervais2016SecurityPerformance, budish2018EconomicLimits,
apostolaki2017HijackingBitcoin}.

The only remaining strategy is splitting the network so that the
recipients of at least two different double-spend transactions commit to
different states. This loss of consistency would materialize in permanent forks
that require out-of-band resolutions (triggered by an else-branch after code line~\ref{l:detectfork}).

In order to understand how \theProtocol{} ensures consistency, it is instructive to
recall the block preference rule in Sect.~\ref{sssec:proto_preference}. Assume
counterfactually that nodes never update their value according to received votes. 
Then, an attacker who becomes the leader could send different proposals to each node. 
This would fragment the honest nodes' compute power and give the attacker time to form
six quorums, three per conflicting state. The probability of the attacker
becoming leader is at least $\alpha$ in each round. %
This would be a catastrophic attack.

The actual block preference rule selects the value with the highest progress among all
known proposals. Therefore, as soon as the first vote is received from an honest node,
all honest nodes converge to a single value. As a result, the
attacker would have to form six complete quorums in the time the honest nodes
get assigned a single ATV and broadcast the corresponding vote. Since $\frac{6 \qsize}{\alpha} \gg \frac1{1-\alpha}$, such an attack becomes infeasible for large quorum sizes \qsize{} and $\alpha<1/2$.

\subsubsection{Censoring} \label{sssec:censoring}

In the censoring scenario, the attacker wants to control the values on which
consensus is achieved for some time. This means he has to be
elected as leader in multiple ($m$) consecutive blocks.

We start with the probability of an attacker becoming the leader in a single round.
Without deviating from the protocol, he leads with probability~$\alpha$. 
This means he could successfully censor \theProtocol{} for $m$ consecutive blocks with probability~$\alpha^m$.

However, naively following the protocol is not the best censoring strategy.
Taking inspiration from the work on selfish mining~\cite{eyal2014MajorityNot,
sapirshtein2016OptimalSelfish, kiayias2016BlockchainMining}, we argue that an attacker can
do better by withholding information. A selfish miner in \nc{} withholds complete blocks, 
such that other miners work on an irrelevant part of the chain. 
\theProtocol{} has a more granular type of information: an attacker might
withhold his votes. A censoring attacker would release his votes only when the
release implies leadership. In practice, this means that a censoring attacker
does not share votes, he only proposes blocks. Using this strategy, the
attacker can delay the next quorum until the honest nodes can form one without
the attacker's votes. This time window increases the attacker's odds of
becoming the leader.

We implement this \emph{censor} strategy and instantiate it in a special
attacker node of the simulation environment (see Fig.~\ref{fig:simulator}). We
bias the assignment of ATVs towards this node such that it posesses
computational power~$\alpha$.
We routinely check for forks, but do not find any.
We count how many of the committed blocks are proposed by the attacker in order
to estimate the probability of leadership per round.
Figure~\ref{fig:leadership} shows this estimate as a function of the quorum
size for different attacker strengths $\alpha$.
Using the described withholding strategy, an $\alpha=1/3$ attacker contributes
roughly 42\,\% ($\alpha=1/2$: 64\,\%) of the blocks. For comparison, the upper bound for block
withholding strategies for the same attacker on \nc{} is 50\,\%
($\alpha=1/2$: 100\,\%)~\cite{sapirshtein2016OptimalSelfish}.

We additionally validate the results on the censor strategy using an independent
Monte Carlo (MC) simulation. (See Appendix~\ref{apx:mcmc} for details.) As depicted in
Figure~\ref{fig:leadership}, the MC analysis confirms the network simulation.

\begin{figure}
  \resizebox{\columnwidth}{!}{
\begin{tikzpicture}[x=1pt,y=1pt]
\definecolor{fillColor}{RGB}{255,255,255}
\path[use as bounding box,fill=fillColor,fill opacity=0.00] (0,0) rectangle (267.40,173.45);
\begin{scope}
\path[clip] (  0.00,  0.00) rectangle (267.40,173.45);
\definecolor{drawColor}{RGB}{0,0,0}

\path[draw=drawColor,line width= 0.4pt,line join=round,line cap=round] ( 45.60, 36.42) -- ( 45.60,167.02);

\path[draw=drawColor,line width= 0.4pt,line join=round,line cap=round] ( 45.60, 36.42) -- ( 39.60, 36.42);

\path[draw=drawColor,line width= 0.4pt,line join=round,line cap=round] ( 45.60, 62.54) -- ( 39.60, 62.54);

\path[draw=drawColor,line width= 0.4pt,line join=round,line cap=round] ( 45.60, 88.66) -- ( 39.60, 88.66);

\path[draw=drawColor,line width= 0.4pt,line join=round,line cap=round] ( 45.60,114.78) -- ( 39.60,114.78);

\path[draw=drawColor,line width= 0.4pt,line join=round,line cap=round] ( 45.60,140.90) -- ( 39.60,140.90);

\path[draw=drawColor,line width= 0.4pt,line join=round,line cap=round] ( 45.60,167.02) -- ( 39.60,167.02);

\node[text=drawColor,anchor=base east,inner sep=0pt, outer sep=0pt, scale=  1.00] at ( 33.60, 32.98) {0.0};

\node[text=drawColor,anchor=base east,inner sep=0pt, outer sep=0pt, scale=  1.00] at ( 33.60, 59.10) {0.2};

\node[text=drawColor,anchor=base east,inner sep=0pt, outer sep=0pt, scale=  1.00] at ( 33.60, 85.22) {0.4};

\node[text=drawColor,anchor=base east,inner sep=0pt, outer sep=0pt, scale=  1.00] at ( 33.60,111.34) {0.6};

\node[text=drawColor,anchor=base east,inner sep=0pt, outer sep=0pt, scale=  1.00] at ( 33.60,137.46) {0.8};

\node[text=drawColor,anchor=base east,inner sep=0pt, outer sep=0pt, scale=  1.00] at ( 33.60,163.58) {1.0};

\path[draw=drawColor,line width= 0.4pt,line join=round,line cap=round] ( 45.60, 31.20) --
	(266.20, 31.20) --
	(266.20,172.25) --
	( 45.60,172.25) --
	( 45.60, 31.20);
\end{scope}
\begin{scope}
\path[clip] (  0.00,  0.00) rectangle (267.40,173.45);
\definecolor{drawColor}{RGB}{0,0,0}

\node[text=drawColor,rotate= 90.00,anchor=base,inner sep=0pt, outer sep=0pt, scale=  1.00] at (  7.20,101.72) {leadership};
\end{scope}
\begin{scope}
\path[clip] (  0.00,  0.00) rectangle (267.40,173.45);
\definecolor{drawColor}{RGB}{0,0,0}

\path[draw=drawColor,line width= 0.4pt,line join=round,line cap=round] ( 53.77, 31.20) -- (258.03, 31.20);

\path[draw=drawColor,line width= 0.4pt,line join=round,line cap=round] ( 53.77, 31.20) -- ( 53.77, 25.20);

\path[draw=drawColor,line width= 0.4pt,line join=round,line cap=round] ( 79.30, 31.20) -- ( 79.30, 25.20);

\path[draw=drawColor,line width= 0.4pt,line join=round,line cap=round] (104.83, 31.20) -- (104.83, 25.20);

\path[draw=drawColor,line width= 0.4pt,line join=round,line cap=round] (130.37, 31.20) -- (130.37, 25.20);

\path[draw=drawColor,line width= 0.4pt,line join=round,line cap=round] (155.90, 31.20) -- (155.90, 25.20);

\path[draw=drawColor,line width= 0.4pt,line join=round,line cap=round] (181.43, 31.20) -- (181.43, 25.20);

\path[draw=drawColor,line width= 0.4pt,line join=round,line cap=round] (206.96, 31.20) -- (206.96, 25.20);

\path[draw=drawColor,line width= 0.4pt,line join=round,line cap=round] (232.50, 31.20) -- (232.50, 25.20);

\path[draw=drawColor,line width= 0.4pt,line join=round,line cap=round] (258.03, 31.20) -- (258.03, 25.20);

\node[text=drawColor,anchor=base,inner sep=0pt, outer sep=0pt, scale=  1.00] at ( 53.77,  9.60) {1};

\node[text=drawColor,anchor=base,inner sep=0pt, outer sep=0pt, scale=  1.00] at ( 79.30,  9.60) {2};

\node[text=drawColor,anchor=base,inner sep=0pt, outer sep=0pt, scale=  1.00] at (104.83,  9.60) {4};

\node[text=drawColor,anchor=base,inner sep=0pt, outer sep=0pt, scale=  1.00] at (130.37,  9.60) {8};

\node[text=drawColor,anchor=base,inner sep=0pt, outer sep=0pt, scale=  1.00] at (181.43,  9.60) {32};

\node[text=drawColor,anchor=base,inner sep=0pt, outer sep=0pt, scale=  1.00] at (206.96,  9.60) {64};

\node[text=drawColor,anchor=base,inner sep=0pt, outer sep=0pt, scale=  1.00] at (232.50,  9.60) {128};

\node[text=drawColor,anchor=base,inner sep=0pt, outer sep=0pt, scale=  1.00] at (258.03,  9.60) {256};
\end{scope}
\begin{scope}
\path[clip] (  0.00,  0.00) rectangle (267.40,173.45);
\definecolor{drawColor}{RGB}{0,0,0}

\node[text=drawColor,anchor=base,inner sep=0pt, outer sep=0pt, scale=  1.00] at (155.90,  3.60) {$k$};
\end{scope}
\begin{scope}
\path[clip] ( 45.60, 31.20) rectangle (266.20,172.25);
\definecolor{drawColor}{RGB}{255,177,192}

\path[draw=drawColor,line width= 0.4pt,dash pattern=on 1pt off 3pt ,line join=round,line cap=round] ( 53.77, 39.04) --
	( 79.30, 39.05) --
	(104.83, 39.06) --
	(130.37, 39.10) --
	(155.90, 39.10) --
	(181.43, 39.09) --
	(206.96, 39.07) --
	(232.50, 39.09) --
	(258.03, 39.10);
\definecolor{drawColor}{RGB}{197,78,109}

\path[draw=drawColor,line width= 0.4pt,line join=round,line cap=round] ( 53.77, 39.06) --
	( 79.30, 39.16) --
	(104.83, 38.97) --
	(130.37, 39.12) --
	(155.90, 39.08) --
	(181.43, 39.06) --
	(206.96, 39.11) --
	(232.50, 38.92) --
	(258.03, 39.13);

\path[draw=drawColor,line width= 0.4pt,dash pattern=on 4pt off 4pt ,line join=round,line cap=round] ( 53.77, 38.92) --
	( 79.30, 39.12) --
	(104.83, 39.07) --
	(130.37, 39.01) --
	(155.90, 39.13) --
	(181.43, 39.09) --
	(206.96, 39.17) --
	(232.50, 38.94) --
	(258.03, 39.10);
\definecolor{drawColor}{RGB}{216,198,131}

\path[draw=drawColor,line width= 0.4pt,dash pattern=on 1pt off 3pt ,line join=round,line cap=round] ( 53.77, 49.53) --
	( 79.30, 49.88) --
	(104.83, 50.22) --
	(130.37, 50.45) --
	(155.90, 50.60) --
	(181.43, 50.65) --
	(206.96, 50.75) --
	(232.50, 50.76) --
	(258.03, 50.73);
\definecolor{drawColor}{RGB}{143,118,0}

\path[draw=drawColor,line width= 0.4pt,line join=round,line cap=round] ( 53.77, 49.68) --
	( 79.30, 49.69) --
	(104.83, 49.47) --
	(130.37, 49.63) --
	(155.90, 49.48) --
	(181.43, 49.58) --
	(206.96, 49.24) --
	(232.50, 49.58) --
	(258.03, 49.62);

\path[draw=drawColor,line width= 0.4pt,dash pattern=on 4pt off 4pt ,line join=round,line cap=round] ( 53.77, 49.58) --
	( 79.30, 50.09) --
	(104.83, 50.42) --
	(130.37, 50.22) --
	(155.90, 50.80) --
	(181.43, 50.85) --
	(206.96, 50.66) --
	(232.50, 50.78) --
	(258.03, 50.45);
\definecolor{drawColor}{RGB}{131,216,167}

\path[draw=drawColor,line width= 0.4pt,dash pattern=on 1pt off 3pt ,line join=round,line cap=round] ( 53.77, 62.56) --
	( 79.30, 64.10) --
	(104.83, 65.27) --
	(130.37, 66.06) --
	(155.90, 66.50) --
	(181.43, 66.97) --
	(206.96, 67.01) --
	(232.50, 67.03) --
	(258.03, 67.02);
\definecolor{drawColor}{RGB}{0,143,61}

\path[draw=drawColor,line width= 0.4pt,line join=round,line cap=round] ( 53.77, 62.37) --
	( 79.30, 62.64) --
	(104.83, 62.31) --
	(130.37, 62.90) --
	(155.90, 62.62) --
	(181.43, 62.30) --
	(206.96, 62.73) --
	(232.50, 62.55) --
	(258.03, 62.38);

\path[draw=drawColor,line width= 0.4pt,dash pattern=on 4pt off 4pt ,line join=round,line cap=round] ( 53.77, 62.33) --
	( 79.30, 63.59) --
	(104.83, 64.91) --
	(130.37, 66.28) --
	(155.90, 66.35) --
	(181.43, 66.91) --
	(206.96, 67.06) --
	(232.50, 67.27) --
	(258.03, 66.84);
\definecolor{drawColor}{RGB}{119,211,236}

\path[draw=drawColor,line width= 0.4pt,dash pattern=on 1pt off 3pt ,line join=round,line cap=round] ( 53.77, 80.02) --
	( 79.30, 83.88) --
	(104.83, 86.74) --
	(130.37, 88.70) --
	(155.90, 89.88) --
	(181.43, 90.41) --
	(206.96, 90.72) --
	(232.50, 90.81) --
	(258.03, 90.98);
\definecolor{drawColor}{RGB}{0,140,176}

\path[draw=drawColor,line width= 0.4pt,line join=round,line cap=round] ( 53.77, 80.04) --
	( 79.30, 80.03) --
	(104.83, 79.87) --
	(130.37, 79.61) --
	(155.90, 79.75) --
	(181.43, 79.70) --
	(206.96, 80.04) --
	(232.50, 79.50) --
	(258.03, 79.65);

\path[draw=drawColor,line width= 0.4pt,dash pattern=on 4pt off 4pt ,line join=round,line cap=round] ( 53.77, 79.56) --
	( 79.30, 83.40) --
	(104.83, 86.69) --
	(130.37, 88.72) --
	(155.90, 89.56) --
	(181.43, 90.29) --
	(206.96, 90.34) --
	(232.50, 90.91) --
	(258.03, 91.04);
\definecolor{drawColor}{RGB}{221,184,249}

\path[draw=drawColor,line width= 0.4pt,dash pattern=on 1pt off 3pt ,line join=round,line cap=round] ( 53.77,101.81) --
	( 79.30,109.04) --
	(104.83,114.53) --
	(130.37,117.57) --
	(155.90,119.20) --
	(181.43,119.98) --
	(206.96,120.44) --
	(232.50,120.59) --
	(258.03,120.70);
\definecolor{drawColor}{RGB}{158,89,199}

\path[draw=drawColor,line width= 0.4pt,line join=round,line cap=round] ( 53.77,101.89) --
	( 79.30,102.22) --
	(104.83,101.80) --
	(130.37,101.95) --
	(155.90,101.90) --
	(181.43,102.09) --
	(206.96,101.34) --
	(232.50,102.03) --
	(258.03,101.93);

\path[draw=drawColor,line width= 0.4pt,dash pattern=on 4pt off 4pt ,line join=round,line cap=round] ( 53.77,102.09) --
	( 79.30,107.52) --
	(104.83,112.05) --
	(130.37,115.09) --
	(155.90,116.84) --
	(181.43,118.01) --
	(206.96,119.22) --
	(232.50,119.11) --
	(258.03,120.14);
\definecolor{drawColor}{RGB}{0,0,0}

\path[draw=drawColor,line width= 0.4pt,line join=round,line cap=round] ( 45.60,172.25) rectangle (140.89,124.25);
\definecolor{drawColor}{gray}{0.28}

\path[draw=drawColor,line width= 0.4pt,line join=round,line cap=round] ( 54.60,160.25) -- ( 72.60,160.25);

\path[draw=drawColor,line width= 0.4pt,dash pattern=on 4pt off 4pt ,line join=round,line cap=round] ( 54.60,148.25) -- ( 72.60,148.25);
\definecolor{drawColor}{gray}{0.67}

\path[draw=drawColor,line width= 0.4pt,dash pattern=on 1pt off 3pt ,line join=round,line cap=round] ( 54.60,136.25) -- ( 72.60,136.25);
\definecolor{drawColor}{RGB}{0,0,0}

\node[text=drawColor,anchor=base west,inner sep=0pt, outer sep=0pt, scale=  1.00] at ( 81.60,156.80) {naive};

\node[text=drawColor,anchor=base west,inner sep=0pt, outer sep=0pt, scale=  1.00] at ( 81.60,144.80) {censor};

\node[text=drawColor,anchor=base west,inner sep=0pt, outer sep=0pt, scale=  1.00] at ( 81.60,132.80) {censor (MC)};
\end{scope}
\end{tikzpicture}}

  \caption{The attacker's share of committed block as a function of
    quorum size for $\alpha \in \left\{\frac{1}{50}, \frac{1}{10}, \frac{1}{5},
    \frac{1}{3}, \frac{1}{2}\right\}$ (bottom-up) in two independent
    simulations (network and MC).}
  \label{fig:leadership}
\end{figure}
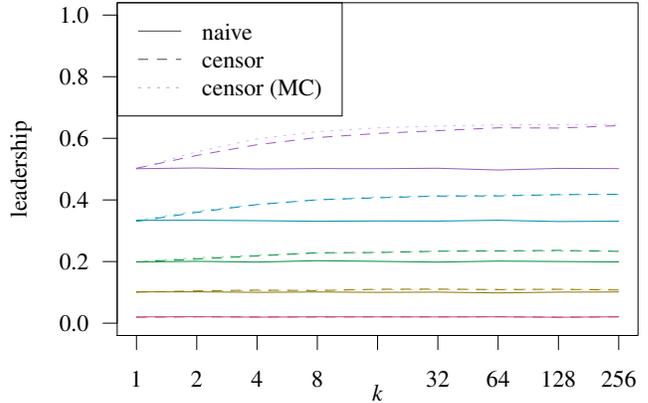

\subsubsection{Chain Quality and Incentive Compatibility} The prevalent strategy
for increasing the own share of blocks and rewards is selfish
mining~\cite{eyal2014MajorityNot, sapirshtein2016OptimalSelfish, negy2020SelfishMining}.
This attack is inherently connected with
incentives.  Its basic idea is to withhold and strategically release blocks in
order to create an information asymmetry that allows to reap a disproportional
amount of rewards for the invested share of work.  This idea is not directly
transferrable from \nc{} to \theProtocol{} for three reasons.
First, the finality after three blocks substantially limits the horizon of the
selfish miner.
Second, block proposals are less valuable. They are not significant sources of
reward.
Third, block proposals are less critical. In fact, block withholding reduces to
the situation of leader failure. Since votes can be reused, honest nodes can
replace missing proposals very fast (see Section~\ref{sssec:leader-failure}).
This makes proposals less rare events than in \nc{}, limiting the
strategic advantage of withholding them.

However, as we have argued in Section~\ref{sssec:censoring}, it is a valid
strategy to \emph{withhold votes}. Therefore, we analyze the effect of vote
withholding on the distribution of rewards, assuming a constant reward per
committed vote, like in Bobtail~\cite{bissias2020BobtailImproved}.
The naive strategy yields a share of~$\alpha$ of the votes. The attacker's goal
is to maximize the number of votes he contributes to each quorum.  Since only
the leader can decide which votes are included in a proposed quorum, the first
step of optimal vote withholding is to increase the odds of becoming the
leader. This, in turn, can be achieved by withholding votes! The circularity
indicates that the attack can be %
approximated with the censoring strategy discussed in
Section~\ref{sssec:censoring}.

Figure~\ref{fig:votes} shows simulation results on how the strategy, $\alpha$,
and the quorum size affect the share of attacker votes committed to the chain.
Interestingly, the censor receives fewer rewards than honest nodes and naive
attackers, indicating a dilemma between paying for becoming the leader and
capitalizing the power of leadership. The tradeoff is visible by comparing
Figures~\ref{fig:leadership} and \ref{fig:votes}.
A similar tradeoff appears for the so-called ``proof withholding''
strategy in Bobtail~\cite{bissias2020BobtailImproved}, which resembles the censoring strategy in \theProtocol{}.

Again, we compare the protocol implementation in the network simulation with
the idealized MC model described in Appendix~\ref{apx:mcmc}.

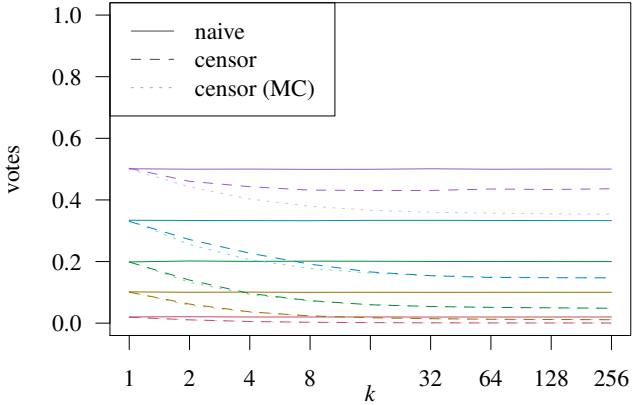
\begin{figure}
  \resizebox{\columnwidth}{!}{
\begin{tikzpicture}[x=1pt,y=1pt]
\definecolor{fillColor}{RGB}{255,255,255}
\path[use as bounding box,fill=fillColor,fill opacity=0.00] (0,0) rectangle (267.40,173.45);
\begin{scope}
\path[clip] (  0.00,  0.00) rectangle (267.40,173.45);
\definecolor{drawColor}{RGB}{0,0,0}

\path[draw=drawColor,line width= 0.4pt,line join=round,line cap=round] ( 45.60, 36.42) -- ( 45.60,167.02);

\path[draw=drawColor,line width= 0.4pt,line join=round,line cap=round] ( 45.60, 36.42) -- ( 39.60, 36.42);

\path[draw=drawColor,line width= 0.4pt,line join=round,line cap=round] ( 45.60, 62.54) -- ( 39.60, 62.54);

\path[draw=drawColor,line width= 0.4pt,line join=round,line cap=round] ( 45.60, 88.66) -- ( 39.60, 88.66);

\path[draw=drawColor,line width= 0.4pt,line join=round,line cap=round] ( 45.60,114.78) -- ( 39.60,114.78);

\path[draw=drawColor,line width= 0.4pt,line join=round,line cap=round] ( 45.60,140.90) -- ( 39.60,140.90);

\path[draw=drawColor,line width= 0.4pt,line join=round,line cap=round] ( 45.60,167.02) -- ( 39.60,167.02);

\node[text=drawColor,anchor=base east,inner sep=0pt, outer sep=0pt, scale=  1.00] at ( 33.60, 32.98) {0.0};

\node[text=drawColor,anchor=base east,inner sep=0pt, outer sep=0pt, scale=  1.00] at ( 33.60, 59.10) {0.2};

\node[text=drawColor,anchor=base east,inner sep=0pt, outer sep=0pt, scale=  1.00] at ( 33.60, 85.22) {0.4};

\node[text=drawColor,anchor=base east,inner sep=0pt, outer sep=0pt, scale=  1.00] at ( 33.60,111.34) {0.6};

\node[text=drawColor,anchor=base east,inner sep=0pt, outer sep=0pt, scale=  1.00] at ( 33.60,137.46) {0.8};

\node[text=drawColor,anchor=base east,inner sep=0pt, outer sep=0pt, scale=  1.00] at ( 33.60,163.58) {1.0};

\path[draw=drawColor,line width= 0.4pt,line join=round,line cap=round] ( 45.60, 31.20) --
	(266.20, 31.20) --
	(266.20,172.25) --
	( 45.60,172.25) --
	( 45.60, 31.20);
\end{scope}
\begin{scope}
\path[clip] (  0.00,  0.00) rectangle (267.40,173.45);
\definecolor{drawColor}{RGB}{0,0,0}

\node[text=drawColor,rotate= 90.00,anchor=base,inner sep=0pt, outer sep=0pt, scale=  1.00] at (  7.20,101.72) {votes};
\end{scope}
\begin{scope}
\path[clip] (  0.00,  0.00) rectangle (267.40,173.45);
\definecolor{drawColor}{RGB}{0,0,0}

\path[draw=drawColor,line width= 0.4pt,line join=round,line cap=round] ( 53.77, 31.20) -- (258.03, 31.20);

\path[draw=drawColor,line width= 0.4pt,line join=round,line cap=round] ( 53.77, 31.20) -- ( 53.77, 25.20);

\path[draw=drawColor,line width= 0.4pt,line join=round,line cap=round] ( 79.30, 31.20) -- ( 79.30, 25.20);

\path[draw=drawColor,line width= 0.4pt,line join=round,line cap=round] (104.83, 31.20) -- (104.83, 25.20);

\path[draw=drawColor,line width= 0.4pt,line join=round,line cap=round] (130.37, 31.20) -- (130.37, 25.20);

\path[draw=drawColor,line width= 0.4pt,line join=round,line cap=round] (155.90, 31.20) -- (155.90, 25.20);

\path[draw=drawColor,line width= 0.4pt,line join=round,line cap=round] (181.43, 31.20) -- (181.43, 25.20);

\path[draw=drawColor,line width= 0.4pt,line join=round,line cap=round] (206.96, 31.20) -- (206.96, 25.20);

\path[draw=drawColor,line width= 0.4pt,line join=round,line cap=round] (232.50, 31.20) -- (232.50, 25.20);

\path[draw=drawColor,line width= 0.4pt,line join=round,line cap=round] (258.03, 31.20) -- (258.03, 25.20);

\node[text=drawColor,anchor=base,inner sep=0pt, outer sep=0pt, scale=  1.00] at ( 53.77,  9.60) {1};

\node[text=drawColor,anchor=base,inner sep=0pt, outer sep=0pt, scale=  1.00] at ( 79.30,  9.60) {2};

\node[text=drawColor,anchor=base,inner sep=0pt, outer sep=0pt, scale=  1.00] at (104.83,  9.60) {4};

\node[text=drawColor,anchor=base,inner sep=0pt, outer sep=0pt, scale=  1.00] at (130.37,  9.60) {8};

\node[text=drawColor,anchor=base,inner sep=0pt, outer sep=0pt, scale=  1.00] at (181.43,  9.60) {32};

\node[text=drawColor,anchor=base,inner sep=0pt, outer sep=0pt, scale=  1.00] at (206.96,  9.60) {64};

\node[text=drawColor,anchor=base,inner sep=0pt, outer sep=0pt, scale=  1.00] at (232.50,  9.60) {128};

\node[text=drawColor,anchor=base,inner sep=0pt, outer sep=0pt, scale=  1.00] at (258.03,  9.60) {256};
\end{scope}
\begin{scope}
\path[clip] (  0.00,  0.00) rectangle (267.40,173.45);
\definecolor{drawColor}{RGB}{0,0,0}

\node[text=drawColor,anchor=base,inner sep=0pt, outer sep=0pt, scale=  1.00] at (155.90,  3.60) {$k$};
\end{scope}
\begin{scope}
\path[clip] ( 45.60, 31.20) rectangle (266.20,172.25);
\definecolor{drawColor}{RGB}{255,177,192}

\path[draw=drawColor,line width= 0.4pt,dash pattern=on 1pt off 3pt ,line join=round,line cap=round] ( 53.77, 39.04) --
	( 79.30, 37.78) --
	(104.83, 37.13) --
	(130.37, 36.81) --
	(155.90, 36.64) --
	(181.43, 36.56) --
	(206.96, 36.52) --
	(232.50, 36.50) --
	(258.03, 36.49);
\definecolor{drawColor}{RGB}{197,78,109}

\path[draw=drawColor,line width= 0.4pt,line join=round,line cap=round] ( 53.77, 39.06) --
	( 79.30, 39.16) --
	(104.83, 39.03) --
	(130.37, 39.04) --
	(155.90, 39.06) --
	(181.43, 39.02) --
	(206.96, 39.03) --
	(232.50, 39.03) --
	(258.03, 39.04);

\path[draw=drawColor,line width= 0.4pt,dash pattern=on 4pt off 4pt ,line join=round,line cap=round] ( 53.77, 38.92) --
	( 79.30, 37.82) --
	(104.83, 37.13) --
	(130.37, 36.79) --
	(155.90, 36.64) --
	(181.43, 36.56) --
	(206.96, 36.52) --
	(232.50, 36.50) --
	(258.03, 36.49);
\definecolor{drawColor}{RGB}{216,198,131}

\path[draw=drawColor,line width= 0.4pt,dash pattern=on 1pt off 3pt ,line join=round,line cap=round] ( 53.77, 49.53) --
	( 79.30, 44.02) --
	(104.83, 41.09) --
	(130.37, 39.53) --
	(155.90, 38.71) --
	(181.43, 38.29) --
	(206.96, 38.08) --
	(232.50, 37.97) --
	(258.03, 37.92);
\definecolor{drawColor}{RGB}{143,118,0}

\path[draw=drawColor,line width= 0.4pt,line join=round,line cap=round] ( 53.77, 49.68) --
	( 79.30, 49.52) --
	(104.83, 49.58) --
	(130.37, 49.48) --
	(155.90, 49.51) --
	(181.43, 49.48) --
	(206.96, 49.49) --
	(232.50, 49.51) --
	(258.03, 49.50);

\path[draw=drawColor,line width= 0.4pt,dash pattern=on 4pt off 4pt ,line join=round,line cap=round] ( 53.77, 49.58) --
	( 79.30, 44.51) --
	(104.83, 41.22) --
	(130.37, 39.48) --
	(155.90, 38.73) --
	(181.43, 38.32) --
	(206.96, 38.08) --
	(232.50, 37.98) --
	(258.03, 37.88);
\definecolor{drawColor}{RGB}{131,216,167}

\path[draw=drawColor,line width= 0.4pt,dash pattern=on 1pt off 3pt ,line join=round,line cap=round] ( 53.77, 62.56) --
	( 79.30, 53.74) --
	(104.83, 48.62) --
	(130.37, 45.78) --
	(155.90, 44.25) --
	(181.43, 43.52) --
	(206.96, 43.09) --
	(232.50, 42.87) --
	(258.03, 42.76);
\definecolor{drawColor}{RGB}{0,143,61}

\path[draw=drawColor,line width= 0.4pt,line join=round,line cap=round] ( 53.77, 62.37) --
	( 79.30, 62.78) --
	(104.83, 62.60) --
	(130.37, 62.71) --
	(155.90, 62.67) --
	(181.43, 62.57) --
	(206.96, 62.56) --
	(232.50, 62.54) --
	(258.03, 62.53);

\path[draw=drawColor,line width= 0.4pt,dash pattern=on 4pt off 4pt ,line join=round,line cap=round] ( 53.77, 62.33) --
	( 79.30, 54.68) --
	(104.83, 49.02) --
	(130.37, 45.97) --
	(155.90, 44.22) --
	(181.43, 43.48) --
	(206.96, 43.11) --
	(232.50, 42.93) --
	(258.03, 42.74);
\definecolor{drawColor}{RGB}{119,211,236}

\path[draw=drawColor,line width= 0.4pt,dash pattern=on 1pt off 3pt ,line join=round,line cap=round] ( 53.77, 80.02) --
	( 79.30, 69.75) --
	(104.83, 63.28) --
	(130.37, 59.60) --
	(155.90, 57.63) --
	(181.43, 56.55) --
	(206.96, 56.02) --
	(232.50, 55.71) --
	(258.03, 55.62);
\definecolor{drawColor}{RGB}{0,140,176}

\path[draw=drawColor,line width= 0.4pt,line join=round,line cap=round] ( 53.77, 80.04) --
	( 79.30, 79.95) --
	(104.83, 79.89) --
	(130.37, 79.85) --
	(155.90, 80.00) --
	(181.43, 80.01) --
	(206.96, 79.95) --
	(232.50, 79.91) --
	(258.03, 79.88);

\path[draw=drawColor,line width= 0.4pt,dash pattern=on 4pt off 4pt ,line join=round,line cap=round] ( 53.77, 79.56) --
	( 79.30, 71.86) --
	(104.83, 66.08) --
	(130.37, 61.43) --
	(155.90, 58.12) --
	(181.43, 56.55) --
	(206.96, 55.79) --
	(232.50, 55.69) --
	(258.03, 55.63);
\definecolor{drawColor}{RGB}{221,184,249}

\path[draw=drawColor,line width= 0.4pt,dash pattern=on 1pt off 3pt ,line join=round,line cap=round] ( 53.77,101.81) --
	( 79.30, 94.16) --
	(104.83, 89.05) --
	(130.37, 85.97) --
	(155.90, 84.30) --
	(181.43, 83.44) --
	(206.96, 83.03) --
	(232.50, 82.78) --
	(258.03, 82.67);
\definecolor{drawColor}{RGB}{158,89,199}

\path[draw=drawColor,line width= 0.4pt,line join=round,line cap=round] ( 53.77,101.89) --
	( 79.30,101.68) --
	(104.83,101.74) --
	(130.37,101.61) --
	(155.90,101.65) --
	(181.43,101.86) --
	(206.96,101.67) --
	(232.50,101.74) --
	(258.03,101.73);

\path[draw=drawColor,line width= 0.4pt,dash pattern=on 4pt off 4pt ,line join=round,line cap=round] ( 53.77,102.09) --
	( 79.30, 96.56) --
	(104.83, 94.26) --
	(130.37, 92.82) --
	(155.90, 92.65) --
	(181.43, 92.66) --
	(206.96, 93.30) --
	(232.50, 93.07) --
	(258.03, 93.40);
\definecolor{drawColor}{RGB}{0,0,0}

\path[draw=drawColor,line width= 0.4pt,line join=round,line cap=round] ( 45.60,172.25) rectangle (140.89,124.25);
\definecolor{drawColor}{gray}{0.28}

\path[draw=drawColor,line width= 0.4pt,line join=round,line cap=round] ( 54.60,160.25) -- ( 72.60,160.25);

\path[draw=drawColor,line width= 0.4pt,dash pattern=on 4pt off 4pt ,line join=round,line cap=round] ( 54.60,148.25) -- ( 72.60,148.25);
\definecolor{drawColor}{gray}{0.67}

\path[draw=drawColor,line width= 0.4pt,dash pattern=on 1pt off 3pt ,line join=round,line cap=round] ( 54.60,136.25) -- ( 72.60,136.25);
\definecolor{drawColor}{RGB}{0,0,0}

\node[text=drawColor,anchor=base west,inner sep=0pt, outer sep=0pt, scale=  1.00] at ( 81.60,156.80) {naive};

\node[text=drawColor,anchor=base west,inner sep=0pt, outer sep=0pt, scale=  1.00] at ( 81.60,144.80) {censor};

\node[text=drawColor,anchor=base west,inner sep=0pt, outer sep=0pt, scale=  1.00] at ( 81.60,132.80) {censor (MC)};
\end{scope}
\end{tikzpicture}}

  \caption{The attacker's share of committed votes as a function of
    quorum size for $\alpha \in \left\{\frac{1}{50}, \frac{1}{10}, \frac{1}{5},
    \frac{1}{3}, \frac{1}{2}\right\}$ (bottom-up) in two independent
    simulations (network and MC).}
  \label{fig:votes}
\end{figure}

\subsection{Overhead} \label{ssec:overhead}

\nc{} requires one message broadcast per block, namely the block itself,
independent of the number of participating nodes. \theProtocol{} adds~$\qsize$
message broadcasts per block---one for each vote. Votes are much smaller than
blocks. Under the conservative assumptions of 256 bits for block reference and
public key, and 64 bits for the puzzle solution, a vote is 72\,B.%
\footnote{Bitcoin shortens public keys to 160 bits and uses solutions of 32
bits. Its blocks are in the order of 1\,MB.}%

The number of messages is constant in the number of nodes, like in Bitcoin.
However, block headers grow. \theProtocol{} must store the complete quorum with
$\qsize$ puzzle solutions. This overhead matters because the header is
replicated in all nodes that want to verify the blockchain in the future.

Assuming the same vote size and the most robust case analyzed ($\qsize=256$),
the storage overhead is about 10\,kB per block.  This is less than 1\,\% of
Bitcoin's average block size in 2019.
With this choice of~$\qsize$, falsely accepting a quorum as unique is much
less likely than guessing a 128-bit key in one attempt.
Table~\ref{tab:overhead} (in the appendix) shows the storage overhead per block
and the associated probability of ambiguity at expected optimistic quorum
time (Corollary~\ref{cor:poa_at_ev}) for different choices of~$\qsize$.
We argue that the benefits of the protocol outweigh its storage costs and leave
the exploration of compression techniques to future work.

\section{Discussion} \label{sec:discussion}

\subsection{Relation to Other Distributed Logs} \label{ssec:related}

\begin{table}
  \small
  \newcommand{\cm}{\checkmark}
  \newcommand{\sw}{(\cm)}
  \newcommand{\mc}[2]{\multicolumn{#1}{l}{#2}}
  \newcommand{\rh}[1]{\rotatebox{90}{#1}}
  \caption{A comparison of related distributed log protocols.}
  \label{tab:design-space}
  \centering
  \begin{adjustbox}{max width=\linewidth}
    \begin{tabular}{lcccccccc}
      \toprule
                                  & \rh{PBFT} & \rh{HotStuff}  & \rh{Proof-of-Stake}  & \rh{Nakamoto} & \rh{Bitcoin-NG} & \rh{Byzcoin} & \rh{Bobtail}  & \rh{\theProtocol{}} \\
      \midrule
      $\#$ nodes   & 10         & $10^2$              & $10^3$                    & $10^3$            & $10^3$               & $10^3$            & $10^3$             & $10^3$                   \\
      \midrule
      committee                   &           &                & \cm                  & \sw          & \sw             & \cm          & \cm           & \cm                 \\
      \midrule
      permissioned               \\
      \quad - network             & \cm       & \cm           \\
      \quad - committee           &           &                & \cm                 \\
      \midrule
      \mc{9}{resource binding {\scriptsize (see Fig.~\ref{fig:btp_vs_bti})}} \\
      \quad - BTP                 &           &                &                      & \cm          &                 &              & \cm          \\
      \quad - BTI                 &           &                &                      &              & \cm             & \cm          &               & \cm                 \\
      \midrule
      sidechain                   &           &                &                      &              & \cm             & \cm          &               &                     \\
      \midrule
      finality                    & \cm       & \cm            & \cm                  &              &                 & \cm          &               & \cm                 \\
      \bottomrule                \\[-2ex]
     \mc{9}{\scriptsize \sw{}: Bitcoin and Bitcoin-NG use single-node committees.}
    \end{tabular}
  \end{adjustbox}
  \undef{\cm}
  \undef{\sw}
  \undef{\mc}
  \undef{\rh}
\end{table}

New distributed log protocols are proposed almost every month.
We do not claim to know all of them and we do not attempt to provide a complete
map of the design space, since other researchers have specialized on this
task~\cite{bano2019SoKConsensus, cachin2017BlockchainConsensus}.
Instead, we compare \theProtocol{} to some of its closest relatives along selected dimensions (see~Table~\ref{tab:design-space}).

\subsubsection{Number of Nodes} \label{sssec:table_size}
Early BFT protocols were designed for a small number of nodes.
PBFT~\cite{castro2002PracticalByzantine}, for example, is proven secure under
the Byzantine assumptions \ref{bft-network} and \ref{bft-adversary}.
It requires multiple rounds of voting to reach consensus on a single
value. The communication complexity of $O(n^2)$ renders it
impractical for more than a dozen nodes $n$.

HotStuff~\cite{yin2019HotStuffBFT} ensures safety under the same assumptions, 
but increases the rate of confirmed values to
one per round of voting.  Its key idea is to pipeline the commit phases of
iterative consensus (recall Fig.~\ref{fig:pipeline}). Moreover, it reduces
communication complexity to $O(n)$ by routing all communication through a leader.
These two changes make HotStuff practical for larger networks. 
However, all correct nodes actively participate (send messages) for each block.

\subsubsection{Committee} \label{sssec:table_committee}
Protocols designed for even larger scale reduce communication
complexity further by electing committees. Only committee members participate actively. 
All other nodes wait until they become part of a committee.

In \nc, write-access to the ledger is controlled by a proof-of-work puzzle.
In each round, one node -- the finder of the block --
broadcasts a message. Consequently, successful miners can be
interpreted as single-node committees.
In Bobtail~\cite{bissias2020BobtailImproved} and \theProtocol{}, multiple proof-of-work puzzles
are solved per block.  Consequently the committee size is greater
than one.
The committee approach is also followed by proof-of-stake protocols. Here,
committee membership is tied to the possession of transferable digital assets
(stake).

\subsubsection{Permissioned} \label{sssec:table_permissioned}
As stated earlier (Sect.~\ref{sec:pow_quorum}), assumption~\ref{bft-network} can only be
satisfied by restricting access to the network based on identities assigned by
an external identity provider or gatekeeper.  Consequently, protocols relying
on this assumption are permissioned on the network layer. %

Proof-of-stake internalizes the gatekeeping functionality by restricting access
to the committee based on the distribution of stake. While participating as
a node is possible without permission, access to the committee is still
permissioned.

In proof-of-work systems any agent can join and leave the network and has a
(fair) chance of becoming committee member without obtaining permission from a
gatekeeper.%
\footnote{We ignore the role of the supply chain for puzzle solving equipment.}

\subsubsection{Resource Binding} \label{sssec:table_binding}
Proof-of-work can be seen as a commitment of resources to a value. %
Typically, these values are chosen locally on each node. Freshness is
guaranteed by including a reference to recent puzzle solutions in the value.
We distinguish between resources bound to a proposal (BTP) for an
upcoming state update and resources bound to an identifier (BTI) used for entering the committee.

\begin{figure}
  \newcommand{\tevent}[2]{
    \draw (#1, 3pt) node [above, align=center] {#2} -- (#1, -3pt);
  }
  \newcommand{\phase}[2]{
    \draw (#1, -3pt) node [below, align=center] {#2};
  }

  \centering

  Bound to proposal (BTP)\\
  \vspace{0.8em}
  \begin{tikzpicture}[x=\columnwidth/11]
    \footnotesize
    \draw[] (0,0) -- (5.6,0);
    \draw[decoration=zigzag, decorate] (5.6,0) -- (8.4,0);
    \draw[->] (8.4,0) -- (11,0) node [below left] {time};
    \phase{7}{resource binding}
    \tevent{10}{publish}
    \tevent{5.5}{define\\proposal}
    \tevent{8.5}{find\\solution}
  \end{tikzpicture}

  \vspace{1.2em}
  Bound to identifier (BTI)\\
  \vspace{0.8em}
  \begin{tikzpicture}[x=\columnwidth/11]
    \footnotesize
    \draw[] (0,0) -- (1.1,0);
    \draw[decoration=zigzag, decorate] (1.1,0) -- (3.9,0);
    \draw[->] (3.9,0) -- (11,0) node [below left] {time};
    \phase{2.5}{resource binding}
    \tevent{10}{publish}
    \tevent{1}{define\\identifier}
    \tevent{4}{find\\solution}
    \tevent{5.5}{define\\proposal}
  \end{tikzpicture}

  \vspace{0.6em}
  \hfill\tikz\draw[decoration=zigzag, decorate] (0,0) -- (1.5,0)
    node [right] {\footnotesize competition};
  \caption{Resources can be bound to concrete proposals or to
  identifiers, which are later used to sign proposals.}
  \label{fig:btp_vs_bti}
  \undef{\tevent}
  \undef{\phase}
\end{figure}
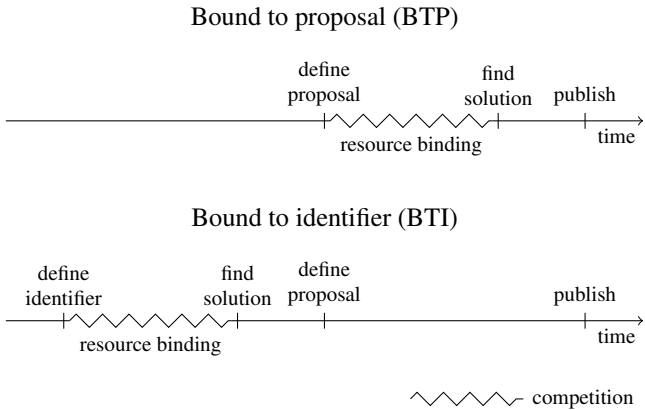

\nc{} uses BTP. Nodes form a proposal for the next block locally and then
start to solve a proof-of-work for this proposal. If they are successful in
finding a puzzle solution, they share their proposal. This
process is depicted in the upper half of Figure~\ref{fig:btp_vs_bti}.

Bitcoin-NG~\cite{eyal2016BitcoinNGScalable} innovated by translating the
concept of leader election from the BFT literature (\eg{}
\cite{dwork1988ConsensusPresence, garciamolina1982ElectionsDistributed,
ongaro2014SearchUnderstandable}) to \nc{}.  The miner of a block (elected
leader) becomes responsible for
appending multiple consecutive (micro) blocks until the next leader emerges with the next mined block. 
In our framework, Bitcoin-NG adds throughput by switching from BTP to BTI in \nc{}.
A more elaborate BTI protocol is Byzcoin~\cite{kogias2016EnhancingBitcoin}. It forms a committee over the
last $\qsize$ successful miners.  This rolling committee is then responsible for appending
micro blocks. Byzcoin uses PBFT to reach final consensus within each committee, thereby shifting control over the micro blocks from a single node (Bitcoin-NG) to multiple nodes.

\theProtocol{} is a BTI protocol: nodes bind resources to identifiers by mining votes. If they happen to lead when the quorum is complete, they sign a block proposal with their secret key.
The lower half of Figure~\ref{fig:btp_vs_bti} shows this order of events.
\annot{Artifact of brainstorm?:
These identifiers can be ephemeral pseudonyms and are not necessarily linkable
to the identity of the agent.} %

Bobtail extends \theProtocol{} by binding a preliminary transaction
list into the proof-of-work solution of each vote.%
\footnote{Since Bobtail inspired \theProtocol{}, a better frame is to see \theProtocol{} as
simplification of Bobtail rather than Bobtail as an extension to \theProtocol{}.}
This BTP aspect of Bobtail adds significant complexity to the voting logic in
order to prevent the reuse of votes for different competing proposals. As described in
Section~\ref{sec:evaluation}, \theProtocol{} makes the reuse of votes a key feature.

\subsubsection{Sidechain} \label{sssec:table_sidechain}

The sequences of micro blocks in Bitcoin-NG, Byzcoin, and also
Thunderella~\cite{pass2018ThunderellaBlockchains} are often referred to as
sidechains. Sidechains can serve several purposes, such as increasing
throughput (Bitcoin-NG) or adding finality (Byzcoin). However, since different
mechanisms are used to advance different chains, synchronization is a major
problem. Bitcoin-NG tackles it with incentives, Thunderella focuses on an
optimistic case, and Byzcoin leaves open which chain has priority. Sidechains
often involve high protocol complexity because different consensus mechanisms
are stacked onto each other: the protocols require a distributed log in order
to provide a distributed log (with different properties). By contrast,
\theProtocol{} provides an improved distributed log directly from a broadcast
network and proof-of-work.

\subsubsection{Finality} \label{sssec:table_finality}

The lack of finality in \nc{} exposes it to many attacks \cite{bonneau2016WhyBuy,gervais2016SecurityPerformance,budish2018EconomicLimits,auer2019DoomsdayEconomics}.
So far, according to conventional wisdom, eventual consistency has been accepted as the price of a truly
permissionless system. Byzcoin challenged this view with a stacked solution involving sidechains. 
\theProtocol{} achieves the same at lower protocol complexity using proof-of-work quorums. Their stochastic uniqueness allows us to transfer the commit process from the permissioned world to the permissionless.

\subsection{Other Related Protocols}

Not included in Table~\ref{tab:design-space} are protocol proposals that
replace the linear data structure of the distributed log with more general
directed acyclic graphs (DAGs)~\cite{sompolinsky2015SecureHighRate,
sompolinsky2016SPECTREFast}. This promises higher scalability
and faster first confirmation in latent networks at the cost of additional
complexity on the application layer, which cannot rely on the total order and
uniqueness of state updates anymore. Also
Fruitchain~\cite{pass2017FruitChainsFair} can be interpreted as a DAG: it
recognizes solutions to hard and easy puzzles but hides the DAG's complexity
from the application layer by not allowing `fruits' to carry state updates.

An even more radical approach is to drop the distributed log completely and
implement a digital asset directly on a secure (source-ordered) broadcast
without consensus~\cite{guerraoui2019ConsensusNumber}. However,
this approach restricts the versatility of the application layer. 
For example, arbitrary smart contract logic is not supported.

\subsection{Limitations and Future Work} \label{ssec:limits}

We have presented a protocol that achieves finality in a permissionless setting
under axiomatic exclusion of the failure modes
\ref{pow-network} and \ref{pow-adversary}, and the acceptance of a negligible failure
probability. The assumption on
\ref{pow-network} and \ref{pow-adversary} are also made for security proofs
of \nc{}~\cite{garay2015BitcoinBackbone,pass2017AnalysisBlockchain}. 
Nevertheless, it is worth discussing their suitability.

Excluding \ref{pow-network} corresponds to assuming a fixed, network-wide
compute power~$\lambda$. But agents can add and remove nodes at their
willing.  Even if the number of nodes is fixed, the computational power of each
node is not. We observe in practice that a control loop, known as difficulty
adjustment (DA), can compensate changes of~$\lambda$ up to a certain degree.
But ample literature shows that the deployed DA algorithms are not
optimal~\cite{kraft2016DifficultyControl, meshkov2017ShortPaper, fullmer2018AnalysisDifficulty,
hovland2017NonlinearFeedback}, especially in case of sudden changes of~$\lambda$.  We
argue that proof-of-work quorums can support more precise difficulty adjustment
algorithms.  A higher quorum size implies more votes and hence more data points
to inform the algorithm about changes of~$\lambda$.

The same effect can be exploited for detecting network-level attacks, such as
eclipse and splits, more accurately. 
(Appendix~\ref{apx:detect} provides additional details.) 
This is relevant in the context of the CAP theorem~\cite{gilbert2002BrewerConjecture}, which tells us that every
distributed system has to sacrifice one out of consistency, availability and
partition tolerance.  \theProtocol{}, as presented, favors availability over
consistency.  It does not implement a mechanism for detecting network splits,
even though it is possible at high confidence for big quorum sizes. The
trade-off could be changed in favor of consistency.  If a split is detected,
the protocol withholds commits (and may notify the application layer in order
to trigger out-of-band resolutions).

The second failure mode,~\ref{pow-adversary}, can be catastrophic and is hard
to rule out. %
We are not aware of any argument that
bounds $\alpha$ to a constant below 50\,\% for any proof-of-work system. In fact,
>50\,\% attacks have been mounted against smaller instances of \nc{} in
practice~\cite{cryptoslate2019percent51attacks}.

Our network simulation in Section~\ref{sec:evaluation} models exponentially
distributed message propagation times. This distribution puts the system
under pressure, but it is not very realistic. Future work might put the
simulation on a more structured network topology. However, since the literature reports a
significant discrepancy between observed topologies and what cryptocurrencies
are designed for~\cite{delgadosegura2019TxProbeDiscovering, mariem2020Allthat},
it is not obvious what an appropriate topology would look like.

Similarly, we leave unexplored how to disseminate \theProtocol{}'s smaller vote
messages efficiently. Votes easily fit into single Internet packets and their
verification requires only one hash evaluation. It might be possible to improve
vote propagation times using UDP-based structured
broadcast~\cite{rohrer2019KadcastStructured} instead of the gossip broadcast
used in many cryptocurrencies.

Finally, we refrain from designing an incentive mechanism for \theProtocol{} for the reasons stated in Section~\ref{ssec:incentives}.
A principled approach would be to explore reward-optimizing strategies (combined
withholding of votes and blocks) automatically using Markov Decision
Processes~\cite{sapirshtein2016OptimalSelfish, zhang2019LayCommon} or even more
sophisticated Reinforcement Learning
techniques~\cite{hou2019SquirRLAutomating}.

\section{Conclusion} \label{sec:conclusion}

We understand \theProtocol{} as a positive example to support our claim that it is
possible to build permissionless distributed logs \emph{with} finality \emph{directly} 
from proof-of-work. The claim is
tentatively supported (with analysis and simulations) until \theProtocol{} is broken. 
We invite the community to prove our claim wrong, and provide running code online to facilitate this task.%
\repofootnote{}
It is not safe to use this code in  systems dealing with real values.

Regardless of whether our claim is true or false, the identified conflict between inclusiveness and security is instructive, and the
associated theory of quorums on stochastic processes may find applications elsewhere. Since it comprises \nc{} 
as a special case, it also contributes to a better understanding of the role of proof-of-work in known systems that ``work in practice, but [so far] not in theory''~\cite{bonneau2015SoKResearch}.

If our claim holds, we have found a way to build  permissionless distributed logs from proof-of-work that can serve many applications better than existing systems. However, proof-of-work is a very wasteful way of establishing consensus. It should be avoided whenever possible. Only if there is no alternative to proof-of-work, \theProtocol{} should be considered as a replacement for \nc{}.

\bibliography{references}

\appendix

\section{Proofs, Figures, and Visualizations}

\begin{figure*}
  \begin{subfigure}[b]{\linewidth}
    \includegraphics[page=8, width=\linewidth]{figures/timeline.pdf}
    \caption{In Bitcoin, puzzles are solved sequentially. Solutions are bound
      to block proposals, implying exponentially distributed block intervals.}
    \label{fig:timeline_bitcoin}
  \end{subfigure}
  \begin{subfigure}[b]{\linewidth}
    \vspace{2ex}
    \includegraphics[page=31, width=\linewidth]{figures/timeline.pdf}
    \caption{In \theProtocol{}, smaller puzzles are solved in parallel. One of the
      solutions is chosen as leader, the corresponding miner collects the
      quorum ($\qsize=8$ votes) and proposes the next block. Thereby, \theProtocol{} enables
    more regular block intervals and more frequent rewards for miners. }
    \label{fig:timeline_hotpow}
  \end{subfigure}
  \begin{subfigure}[b]{\linewidth}
    \vspace{2ex}
    \centering
    \includegraphics[page=1]{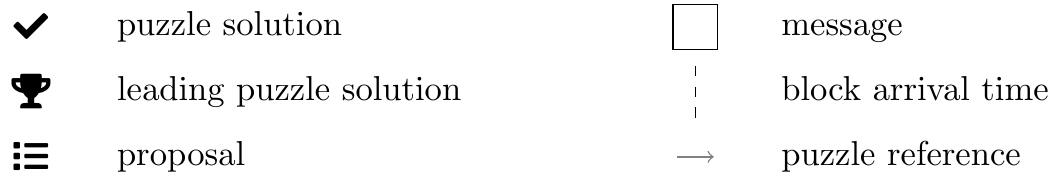}
    \caption{Symbols and their meaning.}
    \label{fig:timeline_legend}
  \end{subfigure}
  \caption{Simulated executions of Bitcoin and \theProtocol{} on $n=7$ nodes
  ($y$-axis) over time ($x$-axis).}
  \label{fig:timeline}
\end{figure*}

\label{apx:proofs}
\paragraph{Lemma \ref{lem:poa_poisson}}
The POA for the Poisson process $P_\lambda$ is given by
\[
  \poa_{P_\lambda, \qsize}(t) =
  1 - e^{-\lambda t} \sum_{i=0}^{2\qsize -1}{\frac{(\lambda t)^i}{i!}} \,.
\]
\begin{proof}
  $P_\lambda$ has the following
  properties~\cite[p.~389]{stewart2009ProbabilityMarkov}:
  \begin{enumerate}
    \item $\prob{P_\lambda(0) = 0 } = 1$,
    \item $P_\lambda(t) - P_\lambda(s) \drawn \distPoisson(\lambda
      \cdot(t-s))$ for all $s < t$, and
    \item for $n \in \nats$ and $0 < t_1 < \dots < t_n$, the family of random
      variables \[\{P_\lambda(t_i) - P_\lambda(t_{i-1})\mid 2 \leq i \leq
      n\}\] is stochastically independent.
  \end{enumerate}

  According to Definition~\ref{def:poa},
  \begin{align}
    \poa_{P_\lambda, \qsize}(t)
    = & \prob{P_\lambda(t) \geq 2 \qsize} \\
    = & 1 - \prob{P_\lambda(t) \leq 2 \qsize - 1} \,.
  \end{align}

  By setting $s=0$ in property 2 of the Poisson process and using property 1,
  we conclude that $P_\lambda(t) \drawn \distPoisson(\lambda t)$. By
  evaluating the cumulative distribution function of the Poisson distribution
  \begin{align}
    F_{\distPoisson}(n; \lambda') =
    e^{-\lambda'} \sum_{i=0}^{\lfloor n \rfloor}{\frac{\lambda'^i}{i!}}
  \end{align}
  for $n = 2\qsize - 1$ and $\lambda' = \lambda t$, we obtain the stated result.
\end{proof}

\paragraph{Lemma~\ref{lem:qt_possion}}
The optimistic $\qsize$-quorum time for the Poisson process is Erlang
distributed with shape parameter~$\qsize$ and rate parameter~$\lambda$, in
short
\[
  T_{P_\lambda,\qsize} \drawn \distErlang(\qsize, \lambda) \,.
\]

\begin{proof}
  The time between two consecutive count events of $P_\lambda$ is exponentially
  distributed with rate parameter $\lambda$. The times between any two
  consecutive count events are stochastically independent. The sum of $\qsize$
  independent and identically distributed exponential random variables is
  Erlang distributed~\cite[p.~146]{stewart2009ProbabilityMarkov} with shape
  parameter~$\qsize$ and rate parameter~$\lambda$.
\end{proof}

\paragraph{Theorem\,~\ref{thm:negligible}}
For the Poisson process, the probability of ambiguity at the expected quorum
time is negligible in the quorum size~$\qsize$.

\begin{proof}
  Let
  \begin{align}
    f(k) :=
    \poa_{P_\lambda, \qsize}(\tEv) =
    1 - e^{-\qsize} \sum_{i=0}^{2\qsize -1}{\frac{\qsize^i}{i!}} \,.
  \end{align}
  Our first observation is that $f(k)$ can be expressed in terms of the
  regularized incomplete Gamma function $P(\alpha, k)$. According to
  \dlmf{8.4.E9},
  \begin{align}
    f(k) = P(2k,k) \,.
  \end{align}
  Following the definition of the regularized incomplete Gamma function
  (see \dlmf{8.2.E4}), we obtain
  \begin{align}
    f(k) = \frac{\gamma(2k,k)}{(2k-1)!} \,, \label{eq:gamma_frac}
  \end{align}
  with the incomplete Gamma function (see \dlmf{8.2.E1})
  \begin{align}
    \gamma(\alpha,k) = \int_{0}^{k}{t^{\alpha -1}e^{-t} dt}\,.
  \end{align}
  We will prove the theorem by providing an (asymptotic) upper bound for
  $f(k)$ that decreases exponentially in $k$. Stirling's
  Approximation~\cite{robbins1955RemarkStirling} provides a useful lower bound
  for the factorial in the denominator of Equation~\ref{eq:gamma_frac}:
  \begin{align}
    n! \geq \sqrt{2\pi}\,n^{n+\frac{1}{2}}\,e^{-n} \label{eq:stirling_lower}
  \end{align}
  We proceed with an upper bound for the enumerator as follows. Let $g(t) =
  t^{2k-1}e^{-t}$ be the function to integrate for $\alpha = 2k$.  Like for
  integrals in general,
  \begin{align}
    \gamma(2k,k) = \int_{0}^{k}{g(t)\,dt} \leq k \cdot \max_{t \in [0,k]}{g(t)} \,.
  \end{align}
  The derivative of $g$ is $g'(t)=e^{-t}(2k-t-1) t^{2k-2}$. For $t \in [0,k]$ the
  derivative $g'$ is greater than zero. Hence the function $g$ is monotonically
  increasing, the maximum is reached at the end of the interval, and
  \begin{align}
    \gamma(2k,k) \leq k^{2k}e^{-k} \,. \label{eq:enum}
  \end{align}
  Applying Approximations~\ref{eq:stirling_lower} and~\ref{eq:enum} to
  Equation~\ref{eq:gamma_frac}, yields
  \begin{align}
    f(k) &\leq \frac{k^{2k}e^{-k}}{\sqrt{2\pi}\,(2k-1)^{2k-\frac{1}{2}}\,e^{-2k+1}} \\
         &= {\left(\frac{k\sqrt{e}}{2k-1}\right)}^{2k} \sqrt{\frac{2k-1}{2\pi e^2}}
  \end{align}
  Observe that
  \begin{align}
    \limsup_{k\to\infty} \frac{\left(\frac{k\sqrt{e}}{2k-1}\right)^{2k}}{\left(\frac{\sqrt{e}}{2}\right)^{2k}}
    = \limsup_{k\to\infty} \left(\frac{2k}{2k-1}\right)^{2k} = e < \infty\,.
  \end{align}
  Thus,
  \begin{align}
    f(k) = O\left(\frac{e^k}{4^k}\sqrt{k}\right).
  \end{align}
  Since $\sqrt{k} < 1.25^k$ for $k>1$, we can conclude
  \begin{align}
    f(k) &= O\left(\frac{e^k}{4^k} 1.25^k\right) \\
         &= O\left(0.85^k\right)\,.
  \end{align}
\end{proof}

\begin{figure}
  \resizebox{\columnwidth}{!}{
\begin{tikzpicture}[x=1pt,y=1pt]
\definecolor{fillColor}{RGB}{255,255,255}
\path[use as bounding box,fill=fillColor,fill opacity=0.00] (0,0) rectangle (267.40,137.31);
\begin{scope}
\path[clip] (  0.00,  0.00) rectangle (267.40,137.31);
\definecolor{drawColor}{RGB}{0,0,0}

\path[draw=drawColor,line width= 0.4pt,line join=round,line cap=round] ( 45.60, 35.14) -- ( 45.60,132.48);

\path[draw=drawColor,line width= 0.4pt,line join=round,line cap=round] ( 45.60, 35.14) -- ( 39.60, 35.14);

\path[draw=drawColor,line width= 0.4pt,line join=round,line cap=round] ( 45.60, 54.61) -- ( 39.60, 54.61);

\path[draw=drawColor,line width= 0.4pt,line join=round,line cap=round] ( 45.60, 74.08) -- ( 39.60, 74.08);

\path[draw=drawColor,line width= 0.4pt,line join=round,line cap=round] ( 45.60, 93.55) -- ( 39.60, 93.55);

\path[draw=drawColor,line width= 0.4pt,line join=round,line cap=round] ( 45.60,113.01) -- ( 39.60,113.01);

\path[draw=drawColor,line width= 0.4pt,line join=round,line cap=round] ( 45.60,132.48) -- ( 39.60,132.48);

\node[text=drawColor,anchor=base east,inner sep=0pt, outer sep=0pt, scale=  1.00] at ( 33.60, 31.69) {1.0};

\node[text=drawColor,anchor=base east,inner sep=0pt, outer sep=0pt, scale=  1.00] at ( 33.60, 51.16) {1.1};

\node[text=drawColor,anchor=base east,inner sep=0pt, outer sep=0pt, scale=  1.00] at ( 33.60, 70.63) {1.2};

\node[text=drawColor,anchor=base east,inner sep=0pt, outer sep=0pt, scale=  1.00] at ( 33.60, 90.10) {1.3};

\node[text=drawColor,anchor=base east,inner sep=0pt, outer sep=0pt, scale=  1.00] at ( 33.60,109.57) {1.4};

\node[text=drawColor,anchor=base east,inner sep=0pt, outer sep=0pt, scale=  1.00] at ( 33.60,129.04) {1.5};

\path[draw=drawColor,line width= 0.4pt,line join=round,line cap=round] ( 45.60, 31.20) --
	(266.20, 31.20) --
	(266.20,136.11) --
	( 45.60,136.11) --
	( 45.60, 31.20);
\end{scope}
\begin{scope}
\path[clip] (  0.00,  0.00) rectangle (267.40,137.31);
\definecolor{drawColor}{RGB}{0,0,0}

\node[text=drawColor,rotate= 90.00,anchor=base,inner sep=0pt, outer sep=0pt, scale=  1.00] at (  7.20, 83.66) {time to commit};
\end{scope}
\begin{scope}
\path[clip] (  0.00,  0.00) rectangle (267.40,137.31);
\definecolor{drawColor}{RGB}{0,0,0}

\path[draw=drawColor,line width= 0.4pt,line join=round,line cap=round] ( 53.77, 31.20) -- (258.03, 31.20);

\path[draw=drawColor,line width= 0.4pt,line join=round,line cap=round] ( 53.77, 31.20) -- ( 53.77, 25.20);

\path[draw=drawColor,line width= 0.4pt,line join=round,line cap=round] ( 94.62, 31.20) -- ( 94.62, 25.20);

\path[draw=drawColor,line width= 0.4pt,line join=round,line cap=round] (135.47, 31.20) -- (135.47, 25.20);

\path[draw=drawColor,line width= 0.4pt,line join=round,line cap=round] (176.33, 31.20) -- (176.33, 25.20);

\path[draw=drawColor,line width= 0.4pt,line join=round,line cap=round] (217.18, 31.20) -- (217.18, 25.20);

\path[draw=drawColor,line width= 0.4pt,line join=round,line cap=round] (258.03, 31.20) -- (258.03, 25.20);

\node[text=drawColor,anchor=base,inner sep=0pt, outer sep=0pt, scale=  1.00] at ( 53.77,  9.60) {0};

\node[text=drawColor,anchor=base,inner sep=0pt, outer sep=0pt, scale=  1.00] at ( 94.62,  9.60) {0.1};

\node[text=drawColor,anchor=base,inner sep=0pt, outer sep=0pt, scale=  1.00] at (217.18,  9.60) {0.4};

\node[text=drawColor,anchor=base,inner sep=0pt, outer sep=0pt, scale=  1.00] at (258.03,  9.60) {0.5};
\end{scope}
\begin{scope}
\path[clip] (  0.00,  0.00) rectangle (267.40,137.31);
\definecolor{drawColor}{RGB}{0,0,0}

\node[text=drawColor,anchor=base,inner sep=0pt, outer sep=0pt, scale=  1.00] at (155.90,  3.60) {leader failure rate};
\end{scope}
\begin{scope}
\path[clip] ( 45.60, 31.20) rectangle (266.20,136.11);
\definecolor{drawColor}{RGB}{0,0,0}

\path[draw=drawColor,line width= 0.4pt,line join=round,line cap=round] ( 53.77, 35.09) --
	( 57.86, 36.71) --
	( 74.20, 39.73) --
	( 94.62, 47.80) --
	(135.47, 61.46) --
	(176.33, 78.97) --
	(217.18,101.07) --
	(258.03,132.23);

\path[draw=drawColor,line width= 0.4pt,dash pattern=on 4pt off 4pt ,line join=round,line cap=round] ( 53.77, 35.52) --
	( 57.86, 36.22) --
	( 74.20, 36.42) --
	( 94.62, 37.97) --
	(135.47, 41.59) --
	(176.33, 45.77) --
	(217.18, 51.38) --
	(258.03, 59.09);

\path[draw=drawColor,line width= 0.4pt,dash pattern=on 1pt off 3pt ,line join=round,line cap=round] ( 53.77, 35.26) --
	( 57.86, 35.10) --
	( 74.20, 35.57) --
	( 94.62, 35.78) --
	(135.47, 36.97) --
	(176.33, 37.99) --
	(217.18, 39.41) --
	(258.03, 41.43);

\path[draw=drawColor,line width= 0.4pt,dash pattern=on 1pt off 3pt on 4pt off 3pt ,line join=round,line cap=round] ( 53.77, 35.19) --
	( 57.86, 35.19) --
	( 74.20, 35.31) --
	( 94.62, 35.41) --
	(135.47, 35.60) --
	(176.33, 35.72) --
	(217.18, 36.36) --
	(258.03, 36.88);

\path[draw=drawColor,line width= 0.4pt,line join=round,line cap=round] ( 45.60,136.11) rectangle (101.10, 64.11);

\path[draw=drawColor,line width= 0.4pt,line join=round,line cap=round] ( 54.60,112.11) -- ( 72.60,112.11);

\path[draw=drawColor,line width= 0.4pt,dash pattern=on 4pt off 4pt ,line join=round,line cap=round] ( 54.60,100.11) -- ( 72.60,100.11);

\path[draw=drawColor,line width= 0.4pt,dash pattern=on 1pt off 3pt ,line join=round,line cap=round] ( 54.60, 88.11) -- ( 72.60, 88.11);

\path[draw=drawColor,line width= 0.4pt,dash pattern=on 1pt off 3pt on 4pt off 3pt ,line join=round,line cap=round] ( 54.60, 76.11) -- ( 72.60, 76.11);

\node[text=drawColor,anchor=base,inner sep=0pt, outer sep=0pt, scale=  1.00] at ( 73.35,124.11) {$k$};

\node[text=drawColor,anchor=base west,inner sep=0pt, outer sep=0pt, scale=  1.00] at ( 81.60,108.67) {2};

\node[text=drawColor,anchor=base west,inner sep=0pt, outer sep=0pt, scale=  1.00] at ( 81.60, 96.67) {8};

\node[text=drawColor,anchor=base west,inner sep=0pt, outer sep=0pt, scale=  1.00] at ( 81.60, 84.67) {32};

\node[text=drawColor,anchor=base west,inner sep=0pt, outer sep=0pt, scale=  1.00] at ( 81.60, 72.67) {128};
\end{scope}
\end{tikzpicture}}

  \caption{Pure effect of leader failure, \ie{} \emph{without} latency. Supplement to Figure~\ref{fig:failure_real}.}
  \label{fig:failure}
\end{figure}
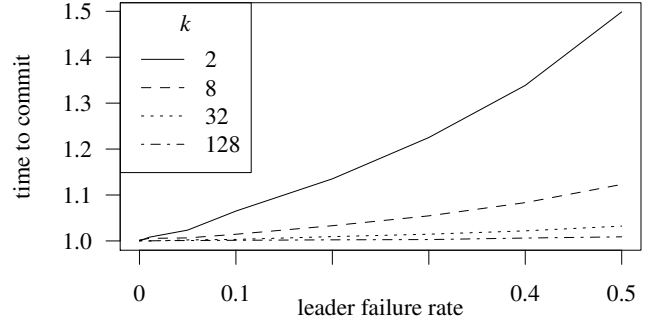

\begin{table}
  \small
  \caption{Storage overhead of \theProtocol{} consensus.}
  \label{tab:overhead}
  \begin{center}
    \begin{tabular}{ccc}
      \toprule
      quorum size & probability of ambiguity & block header
      \\
      \qsize{} & at expected quorum time & (bytes) \\
      \midrule
      1  & $0.2642$ & 72 %
      \\
      2 & $0.1429$ &  112 %
      \\
      16 & $0.0003$  & 672 %
      \\
      64 & $1.2\times 10^{-12}$ & 2.6\,k
      \\
      256 &$4\times 10^{-45}$ & 10\,k
      \\
      \bottomrule
    \end{tabular}
  \end{center}
\end{table}

\section{Monte Carlo Simulation} \label{apx:mcmc}

\begin{figure*}
  \hfil %
  \begin{tikzpicture}[>=stealth,x=15em,y=-7ex]
    \tikzstyle{state}=[]
    \tikzstyle{transition}=[->, draw, rounded corners=1em]
    \tikzstyle{prob}=[anchor=south, font=\small, xshift=-0.5em]
    \tikzstyle{intuition}=[anchor=north, font=\footnotesize, xshift=-0.5em]
    \node[state] (top) at (0,1) {$a,d,\top$};
    \node[state] (ta) at (1,0) {$a+1,d,\top$};
    \node[state] (tb) at (1,1) {$a,d+1,\bot$};
    \node[state] (tc) at (1,2) {$a,d+1,\top$};
    \path[transition] (top) |- (ta);
    \path[transition] (top) -- (tb);
    \path[transition] (top) |- (tc);
    \node[prob] at (0.5, 0) {$\alpha$};
    \node[prob] at (0.5, 1) {$(1 - \alpha)/(a + d + 1)$};
    \node[prob] at (0.5, 2) {$(1 - \alpha)\cdot(a+d)/(a + d + 1)$};
    \node[intuition] at (0.5, 0) {attacker extends lead};
    \node[intuition] at (0.5, 1) {defender obtains lead};
    \node[intuition] at (0.5, 2) {following defender catches up};
  \end{tikzpicture}
  \hfil %
  \hfil %
  \begin{tikzpicture}[>=stealth,x=15em,y=-7ex]
    \tikzstyle{state}=[]
    \tikzstyle{transition}=[->, draw, rounded corners=1em]
    \tikzstyle{prob}=[anchor=south, font=\small, xshift=-0.5em]
    \tikzstyle{intuition}=[anchor=north, font=\footnotesize, xshift=-0.5em]
    \node[state] (top) at (0,1) {$a,d,\bot$};
    \node[state] (ta) at (1,0) {$a,d+1,\bot$};
    \node[state] (tb) at (1,1) {$a+1,d,\top$};
    \node[state] (tc) at (1,2) {$a+1,d,\bot$};
    \path[transition] (top) |- (ta);
    \path[transition] (top) -- (tb);
    \path[transition] (top) |- (tc);
    \node[prob] at (0.5, 0) {$1 - \alpha$};
    \node[prob] at (0.5, 1) {$\alpha/(a + d + 1)$};
    \node[prob] at (0.5, 2) {$\alpha\cdot(a+d)/(a + d + 1)$};
    \node[intuition] at (0.5, 0) {defender extends lead};
    \node[intuition] at (0.5, 1) {attacker obtains lead};
    \node[intuition] at (0.5, 2) {following attacker catches up};
  \end{tikzpicture}
  \hfil %
  \caption{Probabilistic state transitions in the Markov Chain model for the \emph{censor}
  strategy.}
  \label{fig:mcmc}
\end{figure*}
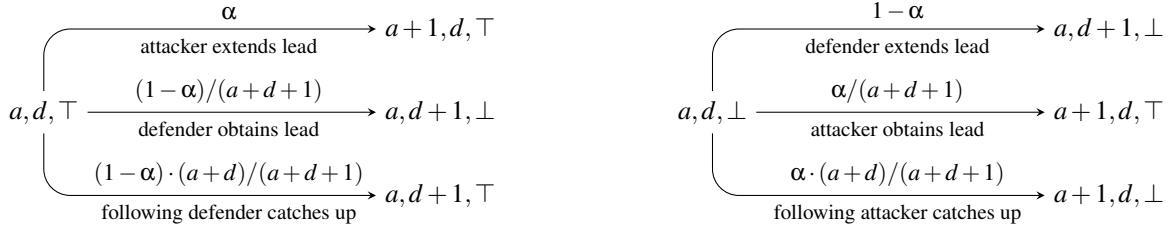

{
  \renewcommand{\wp}{\text{ with probability }}
  \newcommand{\success}{\textsc{success}}
  \newcommand{\fail}{\textsc{fail}}

We cross-check the implementation of the censor strategy and its behavior in
the network simulation (see Sect.~\ref{sssec:censoring}) using an independent
Monte Carlo simulation. We model the formation of individual quorums using
an (Absorbing) Markov Chain, but omit higher-level concepts such as
blocks and their chaining. The censor strategy is to generally withhold
votes until either the attacker can form a quorum as leader, or the defender
forms a quorum without any of the attacker's (withheld) votes. In a protocol
execution, the first case (\success{}) applies when the attacker proposes a block
which the honest nodes accept. The second case
(\fail{}) applies when the honest nodes propose a block.

\paragraph{State representation and initialization}
We model the current state as a triple $(a, d, l)$, where $a \in \nats$ denotes
the number of (withheld) attacker votes, $d \in \nats$ (for defender) denotes the number of
votes of the honest nodes, and $l \in \bools$ is true if the attacker holds the currently
smallest vote. The initial state is $(1, 0, \top)$ with probability $\alpha$
and $(0, 1, \bot)$ otherwise.

\paragraph{State transition}
Figure~\ref{fig:mcmc} shows an annotated state transition diagram.
If $l = \top$, the next state is
\begin{align*}
  &(a + 1, d, l) &&  \wp \alpha ,\\
  &(a, d + 1, \bot) && \wp \frac{1 - \alpha}{a + d + 1} \text{, and}\\
  &(a, d + 1, l) && \text{ otherwise.}
\end{align*}
If $l = \bot$, the next state is
\begin{align*}
  &(a, d + 1, l) && \wp 1 - \alpha ,\\
  &(a + 1, d, \top) && \wp \frac{\alpha}{a + d + 1} \text{, and}\\
  &(a + 1, d, l) && \text{ otherwise.}
\end{align*}

\paragraph{Termination}
If $l \wedge a + d \geq \qsize$, the simulation terminates in \success{}. 
If $\neg l \wedge d \geq \qsize$, it terminates in \fail{}. The simulation continues until one of these conditions is true.

\paragraph{Simulation} We run the model 1\,000\,000 times for each combinations
of $\alpha \in \left\{\frac{1}{50}, \frac{1}{10}, \frac{1}{5}, \frac{1}{3},
\frac{1}{2}\right\}$ and $\qsize \in \left\{1,2,4,\dots,256\right\}$.
Figure~\ref{fig:leadership} shows the fraction of cases where the simulation
terminates in \success{}. Figure~\ref{fig:votes} shows the average number of
attacker votes for the runs that end in \success{}.

} %

\section{Detecting Attacks} \label{apx:detect}

Each vote is linked to one ATV. By assumption (Sect.~\ref{sec:pow_quorum}), the
time between two consecutive ATVs is exponentially distributed with
rate~$\lambda$. In an honest network, a node regularly receives votes (and own
ATVs). A node can test the hypothesis of being eclipsed based on the
arrival of votes.  Table~\ref{tab:detect} shows after how much time (relative
to the block time) of not receiving a single vote a node can rule out a natural
course of events with confidence $p = 0.001$.  

Observe that larger quorums sizes
increase the detectability of eclipse attacks. For quorum sizes greater than 8,
eclipse attacks can be detected with confidence within a single expected block time. 
For plain \nc{} ($\qsize=1$), an equally powerful test requires an observation window of almost 7 times the expected block time.

\begin{table}
  \caption{Time until eclipse can be detected at confidence $p = 0.001$ (relative to the expected block time).}
  \label{tab:detect}
  \resizebox{\linewidth}{!}{
    \begin{tabular}{rccccccccc}
      \toprule
      \textbf{quorum size} & 1    & 2    & 4    & 8    & 16    & 32    & 64    & 128    & 256  \\
      \textbf{time}        & 6.91 & 3.45 & 1.73 & 0.86 & 0.43  & 0.22  & 0.11  & 0.05   & 0.03 \\
      \bottomrule
  \end{tabular}}
\end{table}

\end{document}